\documentclass[journal, onecolumn]{IEEEtran}
\usepackage{algorithmic}
\usepackage{algorithm}
\usepackage{amsmath}
\usepackage{amsfonts}
\usepackage{amssymb}
\usepackage{array}
\usepackage{booktabs}
\usepackage{cite}
\usepackage{CJK}
\usepackage{color}
\usepackage{comment}
\usepackage{epsfig}
\usepackage{epstopdf}
\usepackage{graphicx}
\usepackage{indentfirst}
\usepackage{latexsym}
\usepackage{listings}
\usepackage{pifont}
\usepackage[caption = false]{subfig}
\usepackage{textcomp}
\usepackage{url}
\usepackage{verbatim}
\usepackage{xcolor}
\newtheorem{theorem}{Theorem}

\newtheorem{corollary}[theorem]{Corollary}

\newtheorem{definition}[theorem]{Definition}

\newtheorem{lemma}[theorem]{Lemma}

\newtheorem{problem}[theorem]{Problem}

\newtheorem{remark}[theorem]{Remark}

\newenvironment{proof}[1][Proof]{\emph{#1. }}{\hfill\ensuremath{\blacksquare}}
\begin{document}

\title{Bit Efficient Toeplitz Covariance Estimation}

\author{Hongwei~Xu and Zai~Yang
\thanks{The authors are with the School of Mathematics and Statistics, Xi'an Jiaotong University, Xi'an 710049, China (e-mails: bracy.xu@gmail.com, yangzai@xjtu.edu.cn).}}

\maketitle

\begin{abstract}
This paper addresses the challenge of Toeplitz covariance matrix estimation from partial entries of randomly quantized samples.
To balance the trade-offs among the number of samples, the number of entries observed per sample, and the data resolution, we propose a ruler-based quantized Toeplitz covariance estimator.
We derive non-asymptotic upper and lower bounds for the proposed estimator, and analyze the corresponding convergence rates.
Our results show that the estimator is near-optimal and imply that reducing data resolution within a certain range has limited impact on the estimation accuracy.
Numerical experiments are provided that validate our theoretical findings and show the effectiveness of the proposed estimator.
\end{abstract}


\section{Introduction}
\IEEEPARstart{E}{stimating} the covariance matrix of a distribution $\mathcal{D}$ based on independent random vectors $\boldsymbol{x}^{(1)}, \ldots, \boldsymbol{x}^{(n)} \sim \mathcal{D}$ is a fundamental problem in statistics.
This problem is particularly prominent in applications such as signal processing \cite{krim1996two} and pattern recognition \cite{dahmen2000structured}, where accurate covariance estimation plays a crucial role.
Over the years, significant attention has been devoted to this challenge, resulting in numerous advancements, as evidenced by recent works such as \cite{maly2022new, ke2019user}.

In covariance estimation problems, it is often assumed that the covariance matrix $T$ is a \textit{symmetric Toeplitz matrix}, i.e., it satisfies $T_{a,b} = T_{c,d}$ whenever $|a-b| = |c-d|$.
This assumption naturally arises in stationary processes, where the covariance between measurements at two points depends only on the distance between them.
The Toeplitz structure is particularly prevalent in applications such as direction-of-arrival (DOA) estimation \cite{yang2018sparse}, radar image processing \cite{brookes2008optimising}, and other related program.

Classical covariance estimation faces practical hurdles in implementation.
In practical applications, such as DOA estimation, the estimation of Toeplitz covariance from full-digital, fully observed data requires the deployment of a high-resolution uniform linear array (ULA), in which sensor elements are equally spaced along a straight line.
This configuration faces two major challenges.
On the one hand, constructing a large-scale ULA demands a substantial number of array elements, which not only incurs significant hardware costs \cite{patole2017automotive} but also increases susceptibility to mutual coupling effects \cite{liao2012doa}, thereby degrading estimation accuracy.
To address these issues, sparse linear arrays (SLA) \cite{wu2016direction, li2025sparse} or minimum redundancy arrays (MRA) \cite{moffet1968minimum} are often employed, inevitably leading to observations that contain only a subset of the sample elements.
To explicitly index these sparsely observed positions, we introduce the concept of \textit{ruler} \cite{eldar2020sample}.
On the other hand, in antenna array systems, high-resolution, high-sampling-rate analog-to-digital converters (ADCs) are costly and power-hungry \cite{walden2002analog, hu2025model}.
Employing low-resolution data can substantially enhance cost-effectiveness and energy efficiency while maintaining satisfactory performance, thereby reducing both acquisition precision and sampling overhead.

Motivated by these considerations, we study Toeplitz covariance estimation under \textit{quantized, partially observed data}.
In this paradigm, reducing the number of observed entries per sample or lowering the quantization resolution diminishes the information contained in each observation, and thereby increasing the number of samples needed to achieve a desired estimation accuracy.
Consequently, our objective is to approximate $T$ using as few bits of sample information as possible while satisfying prescribed accuracy constraints, which is particularly critical in large-scale sensing or low-power systems such as radar \cite{patole2017automotive} and wireless communications \cite{studer2016quantized}.
Unlike general covariance estimation, the focus here is on the estimator that leverages only a subset of the entries from each quantized coarse sample $\dot{\boldsymbol{x}}^{(l)}$, see Fig.~\ref{RQTCE} in the main context.

When considering the number of bits required for the estimator $\hat{T}$ to achieve a given tolerance $\epsilon$, i.e.,
\begin{equation}
    \label{Tolerace}
    \|T-\hat{T}\|_2 \leq \epsilon,\qquad \text{or}\qquad \|T-\hat{T}\|_2 \leq \epsilon \|T\|_2,
\end{equation}
the following aspects are typically of interest:
\begin{itemize}
    \item \textbf{Vector Sample Complexity (VSC)}: The number of coarse samples $\dot{\boldsymbol{x}}^{(1)}, \ldots, \dot{\boldsymbol{x}}^{(n)}$ required by the estimator $\hat{T}$ to approximate $T$.
    
    \item \textbf{Entry Sample Complexity (ESC)}: The number of entries observed in each coarse sample.
    
    \item \textbf{Resolution}: Higher resolution is achieved with more precise entries, while lower resolution results from less accurate samples. It is typically determined by the quantization level $\Delta$, where a larger $\Delta$ corresponds to lower resolution.
\end{itemize}
The goal of achieving a covariance estimation that meets the requirement (\ref{Tolerace}) leads to trade-offs among these three factors, as simultaneously minimizing all three is typically infeasible.

The trade-offs between VSC and ESC have been extensively analyzed in \cite{eldar2020sample}, albeit without considering resolution.
The results demonstrate that, even with a relatively small ESC, (\ref{Tolerace}) can still be achieved by increasing the number of observed samples (i.e., VSC).
On the contrary, in general covariance estimation settings, studies have shown that quantization, while reducing resolution, only slightly worsens the leading factors of the estimation error and does not affect the scaling order \cite{chen2023quantizing, chen2023quantized}.

However, to the best of our knowledge, no existing studies simultaneously consider the three factors of VSC, ESC, and resolution.
The aim of this paper is to propose a novel estimator $\hat{T}$ that enables reasonable trade-offs among these three factors, thus advancing the understanding of covariance estimation in Toeplitz settings.

\subsection{Relations to Prior Art}
Estimating the covariance matrix by acquiring all the entries of each sample has been a significant topic in statistics \cite{barton1997structured}.
Classical methods, such as maximum likelihood estimation \cite{burg1982estimation} and projected covariance estimation \cite{roberts1987digital}, have been extensively studied.
In addition, methods that consider only a subset of the entries in each sample have also been explored \cite{gonen2016subspace}.
The introduction of the Toeplitz property into these methods has been shown to significantly enhance estimation performance \cite{miller1987role}.
This improvement arises from the inherent translational invariance of the covariance matrix in stationary processes, which effectively reduces the degrees of freedom required for estimation in practical applications \cite{cai2013optimal}.
Leveraging the Toeplitz property not only improves estimation accuracy but also reduces computational complexity.

To address the Toeplitz covariance estimation problem, sparse recovery methods built on the Vandermonde decomposition of $T$ have been employed \cite{gonen2016subspace, yang2014discretization}.
In particular, several well-known algorithms, such as MUSIC and ESPRIT methods, have been studied under the low-rank setting.
Furthermore, under the strong separation assumption, sparse Fourier transform methods can be applied to solve this problem \cite{eldar2020sample, chen2016fourier}.
In radar and array signal processing, a coherent line of algorithmic methods has been developed for Toeplitz-structured covariance estimation, including projection-based, majorization-minimization-based maximum likelihood estimation (MLE), and expectation-maximizatio (EM) approaches \cite{du2020toeplitz, aubry2024advanced, aubry2021structured, yang2023robust}.
While prior works focus on full-digital (unquantized) covariance estimation and typically rely on refined, high-complexity algorithms to improve performance, this paper investigates the fundamental limits of a simple estimator constructed from quantized measurements by leveraging appropriately designed dithering to ensure desirable statistical properties.

Exploiting the structural properties of the Toeplitz covariance matrix, it is often sufficient to consider only a subset of entries from each sample.
The indices of this subset are referred to as the \textit{ruler} (see Definition~\ref{Ruler}), which is analogous to the sparse linear array in array signal processing.
Consequently, the trade-off between VSC and ESC becomes a critical consideration.
To address this, numerous studies have explored sparse ruler-based methods, which have found widespread applications in signal processing \cite{pillai1985new, romero2015compressive}.
Some results have been established for the ruler-based estimator $\tilde{T}$.
Preliminary analyses of the estimation error $\|T-\tilde{T}\|_2$ are presented in \cite{wu2017toeplitz}.
Further advancements, including non-asymptotic bounds in the sense of the Frobenius norm, have been provided in \cite{qiao2017gridless}.
More recently, non-asymptotic bounds in terms of operator norm have been established in \cite{yang2023nonasymptotic}, further refining the estimation error analysis in Toeplitz matrix recovery.
Groundbreaking results were achieved in \cite{eldar2020sample}, where the relationship between ESC and VSC was comprehensively analyzed in different cases, leading to more general conclusions.

Besides covariance estimation, our research is also related to quantization.
It plays a critical role in reducing the resolution of data, which is crucial for reducing the cost of data processing \cite{chen2023quantizing}.
For instance, in multiple-input multiple-output (MIMO) communication systems, the frequent transmission of high-resolution data imposes significant processing costs.
By quantizing the transmitted data to lower resolutions, these costs can be substantially reduced \cite{mo2018limited}.
Motivated by practical needs, quantization theory has been extensively explored, including 1-bit quantization \cite{dirksen2022covariance, dirksen2021non, jacques2013robust, plan2013one} and multi-bit quantization \cite{chen2023quantizing, jacques2017time, jung2021quantized}. 

Quantization refers to the process of mapping continuous inputs to discrete form \cite{gray1998quantization}.
Various quantization approaches have been proposed to address different problem domains.
For example, \cite{dirksen2022covariance} introduces two methods for 1-bit quantization, which map continuous signals to $\pm 1$.
The first method, dither-free quantization, directly maps data $x$ to $\mathrm{sign}(x)$,
while the second method, dithered quantization, involves adding a uniform dither to the data before quantization.
In addition, \cite{chen2023quantizing} considers a multi-bit dithered quantization approach and analyzes the properties of both uniform and triangular dithers.
Notably, these quantization techniques are memoryless, meaning that the quantization of each data is independent to the others.
In addition to these methods, other quantization approaches have also been studied \cite{dirksen2019quantized}.
The choice of quantization method often depends on the specific problem to be addressed, with different techniques tailored to achieve optimal results in various scenarios.

Research on quantized covariance estimation remains relatively limited.
The problem of 1-bit quantized covariance estimation is explored in \cite{dirksen2022covariance}, where non-asymptotic bounds on the operator norm are analyzed.
The results in \cite{dirksen2022covariance} were shown to be effective for estimating covariance matrices in the masked setting with a known sparse pattern, while the case of an unknown sparse pattern was subsequently addressed in \cite{chen2023high}.
The framework of \cite{dirksen2022covariance} was further extended in \cite{maly2022new}.
Multi-bit quantized covariance estimation is investigated in \cite{chen2023quantizing}, particularly for data with heavy-tailed distributions.
The study derives non-asymptotic bounds of the error on the operator norm and demonstrates that the proposed method achieves near-optimality.
These findings highlight the effectiveness of multi-bit quantization techniques in addressing the challenges posed by heavy-tailed distributions, further enriching the theoretical landscape of quantized covariance estimation.

Recently, further progress has been made in quantized covariance estimation.
In \cite{lu20241}, a 1.5-bit quantization framework was proposed, leading to an arcsin-law-type covariance estimator.
Building upon the results of \cite{dirksen2022covariance}, \cite{dirksen2024tuning} introduced a tuning-free one-bit covariance estimation method, which employs a data-driven approach to automatically determine quantization parameters from the observed data, thereby enabling recovery of the covariance structure without prior knowledge of the true covariance norm.
Moreover, \cite{chen2025parameter} proposed a novel covariance estimation method based on two-bit quantized data, which, in parallel to \cite{dirksen2022covariance}, yields another parameter-free covariance estimator and provides a theoretical analysis of its approximation performance.

\subsection{Our Contributions}
In this paper, we address the problem of quantized Toeplitz covariance estimation under partial observation, which has not been systematically addressed in the literature.
The key contribution is a novel estimator that explicitly exploits the Toeplitz structure under quantization and partial observation, enabling trade-offs among ESC, VSC, and resolution.
We further establish non-asymptotic performance bounds, thereby providing rigorous theoretical guarantees and revealing new insights into the interplay between sample complexity (i.e., VSC and ESC) and quantization resolution.

While Toeplitz covariance estimators with partial observations have been studied in \cite{eldar2020sample} and quantized covariance estimation has been systematically explored in \cite{chen2023quantized, chen2023quantizing, chen2025parameter}, existing approaches typically treat these two aspects separately.
Instead of a simple extension of previous works, our method integrates them into a unified framework, following the standard technique of reducing the operator norm error to a uniform bound on the associated spectral density function.
Specifically, \cite{chen2023quantized} characterizes quantization and estimation errors separately using a random matrix framework, obtaining general results for heavy-tailed data, but without leveraging the structural advantages of Toeplitz matrices, limiting applicability to fully observed settings.
Conversely, \cite{eldar2020sample} exploits Toeplitz structure under partial observation but relies on Gaussian assumptions, which cannot accommodate the loss of Gaussianity in the data induced by quantization, making it unsuitable for quantized data.
By jointly addressing both challenges, our estimator fills this gap and provides a unified solution applicable to scenarios where quantization and partial observation occur simultaneously.

\begin{figure}[!t]
    \centering
    \includegraphics[scale=0.36]{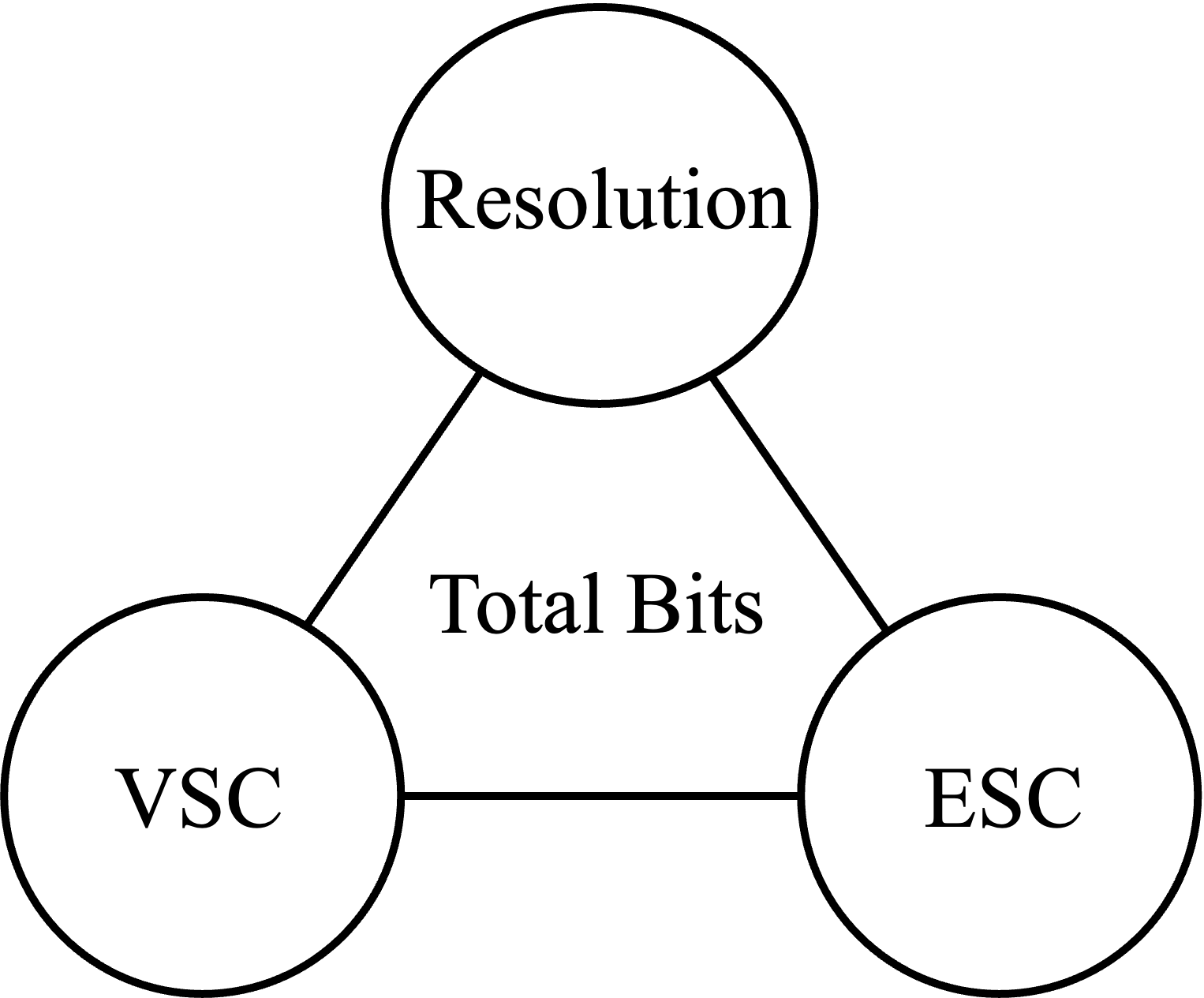}
    \caption{The ``Impossible Triangle" in Toeplitz covariance estimation: for a given tolerance $\epsilon$, to achieve $\|T - \hat{T}\|_2 \leq \epsilon \|T\|_2$, it is impossible to simultaneously minimize VSC, ESC and resolution}
    \label{ImpossibleT}
\end{figure}

Our main contributions in this paper are summarized as follows:

\emph{1) Realized trade-offs between VSC, ESC, and resolution.}
We propose an estimator $\hat{T}$, modified from the ruler-based estimator, that achieves a balanced trade-off among VSC, ESC and resolution.
Specifically, we demonstrate that to achieve $\|T - \hat{T}\|_2 \leq \epsilon \|T\|_2$, it is impossible to simultaneously minimize VSC, ESC, and resolution (e.g., Fig.~\ref{ImpossibleT}).
Furthermore, our results show that decreasing any one (or two) of these quantities can be compensated by increasing the remaining quantity.
We provide an in-depth analysis of the constraints over these three factors.

\emph{2) Non-asymptotic bounds for the estimator $\hat{T}$.}
We derive non-asymptotic bounds that guarantee the performance of the proposed estimator $\hat{T}$.
Our results demonstrate that even when using quantized coarse sample data and observing only a subset of the entries in each sample, the estimator still achieves satisfactory performance.
This highlights the effectiveness of $\hat{T}$ and shows that highly accurate covariance estimation can be achieved using only a limited amount of sample information.

\emph{3) Near-optimal convergence rates.}
We show that quantization only affects the rate of convergence slightly.
Compared to the unquantized results in \cite{eldar2020sample}, our convergence rates differ solely in the coefficient.
This indicates that reducing the sample resolution within a reasonable range has little impact on the estimation accuracy.
Notably, our results degenerate into those in \cite{eldar2020sample} in the case when the data is unquantized.

Our results have significant practical value, offering substantial reductions in data-processing costs and providing insight into the trade-offs and efficiencies inherent in quantized Toeplitz covariance estimation.

\subsection{Notations}

The set of real numbers is denoted by $\mathbb{R}$.
For a positive integer $d$, let $[d] = \{0, 1, \ldots, d-1\}$.
Boldface letters are used to represent vectors, where the $j$-th entry of a vector $\boldsymbol{x}$ is denoted by $x_j$.
Matrices are denoted by regular letters, with $A_{j,k}$ representing the $(j,k)$-th entry of a matrix $A$. For a matrix $A$, we consider its trace $\text{tr}(A)$, rank $\text{rank}(A)$, operator (spectral) norm $\|A\|_2$, Frobenius norm $\|A\|_F$, and maximum norm $\|A\|_{\infty}$.

For any vector $\boldsymbol{a} = (a_0,\ldots, a_{d-1})^T \in \mathbb{R}^d$, let $\text{Toep}(\boldsymbol{a}) \in \mathbb{R}^{d \times d}$ denote the symmetric Toeplitz matrix with entries $\text{Toep}(\boldsymbol{a})_{j,k} = a_{|j-k|}$. We denote $\text{avg}(A)$ as the Toeplitz matrix obtained by averaging the diagonals of $A$.
For a matrix $A \in \mathbb{R}^{d \times d}$ and a set $R \subset [d]$, let $T_R \in \mathbb{R}^{|R| \times |R|}$ represent the principal submatrix of $A$, with rows and columns indexed by $R$, where $|R|$ is the cardinality of $R$.

For a random variable $X$, its expectation is denoted by $\mathbb{E}(X)$. The sub-exponential norm $\|X\|_{\psi_1}$ and sub-Gaussian norm $\|X\|_{\psi_2}$ are defined as
\begin{equation}
    \|X\|_{\psi_1} = \inf \{ t>0 : \mathbb{E} \exp (|X|/t) \leq 2 \}
\end{equation}
and
\begin{equation}
    \|X\|_{\psi_2} = \inf \{ t>0 : \mathbb{E} \exp (X^2/t^2) \leq 2 \},
\end{equation}
respectively. For real numbers $a < b$, we write $X \sim \mathcal{U}([a,b])$ if and only if $X$ follows a uniform distribution over the interval $[a,b]$.

\subsection{Organization}
The rest of the paper is organized as follows.
In Section~\ref{sec:perl}, we revisit the preliminaries that will be used throughout the article.
Section~\ref{sec:rule} then formally defines the ruler-based quantized Toeplitz covariance estimation problem and introduces an unbiased estimator.
A non-asymptotic performance analysis of the proposed estimator is then presented in Section~\ref{sec:nona}.
In Section~\ref{sec:fini}, we develop the finite-bit estimator and provide its non-asymptotic performance guarantees.
Section~\ref{sec:lowr} explores two special cases, i.e., the low-rank and band-limited scenarios.
Numerical experiments validating the theoretical analysis are reported in Section~\ref{sec:nume},
and detailed proofs of the theoretical results are provided in the appendices.

\section{Preliminaries} \label{sec:perl}
In this section, we introduce the key concepts related to ruler-based Toeplitz covariance estimation and quantization.
These concepts serve as foundational knowledge that will be applied in subsequent sections.

\subsection{Ruler-based Toeplitz Covariance Estimation}
Traditional covariance estimation addresses the problem of estimating the covariance matrix $T$ from independent and identically distributed (i.i.d.) samples drawn from a $d$-dimensional normal distribution, where $T \in \mathbb{R}^{d \times d}$ is a positive semidefinite matrix, i.e., $\boldsymbol{x}^{(1)}, \ldots, \boldsymbol{x}^{(n)} \sim \mathcal{N}(0, T)$, given query access to the sample entries.
The objective is to obtain an estimator $\tilde{T}$ that satisfies, with probability at least $1 - \delta$,
\begin{equation}
    \|T-\hat{T}\|_2 \leq \epsilon,\qquad \text{or}\qquad \|T-\hat{T}\|_2 \leq \epsilon \|T\|_2,
\end{equation}
where $\|\cdot\|_2$ denotes the operator norm.

Moreover, when $T$ is a symmetric Toeplitz matrix, it can be fully characterized by a vector $\boldsymbol{a}$, where its element $a_s$ represents the value on the $(s+1)$-th diagonal of $T$.
Specifically, $a_s$ corresponds to the covariance of samples separated by $s$ indices.
Consequently, it suffices to extract a subset of entries $R \subset [d]$ from each sample $\boldsymbol{x} \sim \mathcal{N}(0, T)$ for distances $s = 0, \ldots, d-1$.
This observation implies that covariance estimation can be performed more efficiently under the Toeplitz structure.

In Toeplitz covariance estimation, the goal is to approximate $T$ within a given tolerance $\epsilon$ while only observing a subset of entries from each sample and using as few samples as possible.
To formalize this approach, the concept of a \textit{ruler} is introduced \cite{eldar2020sample}.

\begin{definition}
        \label{Ruler}
        A subset $R \subset [d]$ is called a \textit{ruler} if for all $s = 0, \ldots, d-1$, there exist $j, k \in R$ such that $|j - k| = s$.
        We further let $R_s$ denote the set of ordered pairs $(j, k)$ with distance $s$.
\end{definition}

We say $R$ is \textit{sparse} if $|R| < d$.
For example, when $d = 10$, a sparse ruler can be $R = \{1, 2, 5, 8, 10\}$.
Sparse rulers are also called redundancy sparse arrays or sparse linear arrays in the literature \cite{wu2016direction}.
Note that sparse rulers are not unique.
In general, a sparse ruler must satisfy $\sqrt{d} \leq |R| < d$ \cite{leech1956representation}.
For any $d$, we introduce a special kind of rulers as follows.

\begin{definition}
        \label{SpecialRuler}
        For $\alpha \in [1/2, 1]$, let $R_{\alpha}$ denote a ruler defined as
        \begin{equation}
                R_{\alpha} = R_{\alpha}^{(1)} \cup R_{\alpha}^{(2)},
        \end{equation}
        where
        \begin{equation}
                R_{\alpha}^{(1)} = \{ 1, 2, \ldots, d^{\alpha} \}, \qquad
                R_{\alpha}^{(2)} = \{ d, d - d^{1-\alpha}, \ldots, d - (d^{\alpha} - 1)d^{1-\alpha} \}.
        \end{equation}
\end{definition}

For simplicity, we assume that both $d^{\alpha}$ and $d^{1-\alpha}$ are integers.
In practical applications, rounding non-integer values to the nearest integer is sufficient.
For example, when $d = 16$, the sparse ruler $R_{1/2} = \{1, 2, 3, 4, 8, 12, 16\}$.
As will be demonstrated in subsequent sections, this type of ruler possesses several desirable properties, making it a compelling choice in scenarios where reducing the number of entries observed in each sample is critical.
Notably, when $\alpha = 1$, i.e., $R = [d]$, we refer to it as the \textit{full ruler}, which corresponds to the classic case.

For brevity, we then introduce the concept of the \textit{coverage coefficient} here.

\begin{definition}
        For any ruler $R \subset [d]$, the \textit{coverage coefficient} is defined as
        \begin{equation}
                \phi(R) := \sum_{s=1}^{d-1} \frac{1}{|R_s|}.
        \end{equation}
        where $R_s$ is as specified in Definition~\ref{Ruler}.
\end{definition}

As the distance $s$ is represented more frequently in the ruler $R$, the value of $\phi(R)$ decreases.
The coverage coefficient plays a crucial role in the subsequent proofs.
Intuitively, it can be observed that
\begin{equation}
        \label{RRRRRuler}
        \phi(R_1) = O(\log d), \qquad \phi(R_{1/2}) = O(d).
\end{equation}
In fact, the value of $\phi(R)$ depends on several factors, including the sample dimension $d$, the cardinality of $R$ (i.e., ESC), and the positional structure of the indices in $R$.

A ruler-based Toeplitz covariance estimator can be defined as $\tilde{T} = \text{Toep}(\tilde{\boldsymbol{a}})$, where
\begin{equation}
        \label{RulerEst}
        \tilde{a}_s := \dfrac{1}{n |R_s|} \sum_{l = 1}^n \sum_{(j, k) \in R_s} x_j^{(l)} x_k^{(l)}.
\end{equation}
Here, $R_s$ is defined in Definition~\ref{Ruler}, and $|R_s|$ denotes its cardinality.
In particular, if $R = [d]$, the estimator simplifies to
\begin{equation}
        \tilde{T} = \text{avg} \left(\dfrac{1}{n} \sum_{l = 1}^n \boldsymbol{x}^{(l)} \boldsymbol{x}^{(l)T}\right).
\end{equation}
The non-asymptotic bounds of the estimator $\tilde{T}$ have been extensively studied in \cite{eldar2020sample, qiao2017gridless}.

\subsection{Dithered Quantization Scheme}
Quantization refers to the process of mapping continuous inputs into a discrete form \cite{gray1998quantization}.
This technique is particularly useful for processing continuous data, as it enables the conversion of data into a discrete form with low resolution, thereby significantly reducing the costs associated with data processing and transmission \cite{shlezinger2020uveqfed, zhang2017zipml}.
Consequently, quantization theory has received considerable attention, especially in the fields of statistical learning and estimation \cite{chen2023high}.
A key challenge in quantization lies in navigating the trade-off between achieving high approximation accuracy and minimizing the processing costs.

A common quantization scheme is the \textit{dithered quantization scheme}, which introduces an appropriate random dither to the signal before quantization.

\begin{definition}
        \label{Def_dithered_quantization}
        Let $\boldsymbol{x} \in \mathbb{R}^d$ denote the signal and $\boldsymbol{\tau}\in \mathbb{R}^d$ be a random dither (independent of $\boldsymbol{x}$) with i.i.d. entries sampled from some distribution.
        The dithered quantization is defined as
        \begin{equation}
                \dot{\boldsymbol{x}} = \mathcal{Q}_\Delta(\boldsymbol{x} + \boldsymbol{\tau}),
        \end{equation}
        where
        \begin{equation}
                \mathcal{Q}_\Delta(x) := \Delta \left( \left\lfloor \frac{x}{\Delta} \right\rfloor + \frac{1}{2} \right) \in \Delta \cdot \left( \mathbb{Z} + \frac{1}{2} \right),
        \end{equation}
        and $\mathcal{Q}_\Delta(\boldsymbol{x})$ denotes entry-wise quantization of $\boldsymbol{x}$.
\end{definition}

The dither in Definition~\ref{Def_dithered_quantization} is essential.
Consider two distinct sub-Gaussian random variables $X$ and $Y$ with probability density functions
\begin{equation}
        f_X(x) = 1 - |x|, \quad |x| \leq 1,
        \qquad \text{and}
        \qquad
        f_Y(y) = |y|, \quad |y| \leq 1.
\end{equation}
It is straightforward to verify that $\mathrm{Var}(X) = \frac{1}{6} \neq \frac{1}{2} = \mathrm{Var}(Y)$.
However, under the dither-free uniform quantizer $\mathcal{Q}_1(\cdot)$, we have
\begin{equation}
        \mathbb{P}\left(\mathcal{Q}_1(X) = \dfrac{1}{2}\right) = \mathbb{P}\left(X \geq 0\right) = \dfrac{1}{2}, \qquad \mathbb{P}\left(\mathcal{Q}_1(X) = -\dfrac{1}{2}\right) = \mathbb{P}\left(X < 0\right) = \dfrac{1}{2},
\end{equation}
and hence $\mathrm{Var}(\mathcal{Q}_1(X)) = \frac{1}{4}$.
The same calculation shows $\mathrm{Var}(\mathcal{Q}_1(y)) = \frac{1}{4}$.
This example illustrates that, without dithering, different original variances can collapse to the same quantized variance, making it impossible to recover the true variance from the quantized data.

Within the framework defined in Definition~\ref{Def_dithered_quantization}, several special designations are defined as follows.

\begin{itemize}
        \item{For a given quantization level $\Delta > 0$, the \textit{uniform dither} $\boldsymbol{\tau} = [\tau_i]$ is defined as 
        \begin{equation}
                \tau_i \sim \mathcal{U} \left( \left[ -\frac{\Delta}{2}, \frac{\Delta}{2} \right] \right).
        \end{equation}}
        
        \item{For a given quantization level $\Delta > 0$, the \textit{triangular dither} $\boldsymbol{\tau} = [\tau_i]$ is defined as
        \begin{equation}
                \label{TriangularDither}
                \tau_i \sim \mathcal{U} \left( \left[ -\frac{\Delta}{2}, \frac{\Delta}{2} \right] \right) + \mathcal{U} \left( \left[ -\frac{\Delta}{2}, \frac{\Delta}{2} \right] \right).
        \end{equation}}

        \item{The \textit{quantization error} is defined as the difference between the input and output of the quantizer, i.e.,
        \begin{equation}
                \boldsymbol{\omega} := \dot{\boldsymbol{x}} - (\boldsymbol{x} + \boldsymbol{\tau}).
        \end{equation}
    }

    \item{The \textit{quantization noise} is defined as the overall difference between original input $\boldsymbol{x}$ and final output, i.e.,
    \begin{equation}
        \boldsymbol{\xi} := \dot{\boldsymbol{x}} - \boldsymbol{x}
    \end{equation}
    }
\end{itemize}
In most cases, the distribution of quantization noise $\boldsymbol{\xi}$ is unknown, making it challenging to analyze the effects of quantization.
Fortunately, this issue can be partly addressed by a fundamental result from \cite{gray1993dithered}, stated as follows.

\begin{lemma}
        \label{LemmaOfQuantization}
        Let $\boldsymbol{x} = [x_i]$ be the input signal, and let $\boldsymbol{\tau} = [\tau_i]$ be the random dither, where the entries of $\boldsymbol{\tau}$ are i.i.d. copies of a random variable $Y$. We denote the complex unit by $\text{i}$. Then:
        \begin{itemize}
                \item[(a)] \textit{Quantization Error}:  
                Let $\boldsymbol{\omega} = \dot{\boldsymbol{x}} - (\boldsymbol{x} + \boldsymbol{\tau}) = [\omega_i]$ be the quantization error. If $f(u) := \mathbb{E}(\exp(\text{i} u Y))$ satisfies $f\left( \frac{2 \pi l}{\Delta} \right) = 0$ for all non-zero integers $l$, then:
                \begin{enumerate}
                        \item $x_i$ and $\omega_j$ are independent for all $i, j \in [d]$,
                        \item $\{\omega_j : j \in [d]\}$ are i.i.d. with distribution $\mathcal{U} \left( \left[ -\frac{\Delta}{2}, \frac{\Delta}{2} \right] \right)$.
                \end{enumerate}

                \item[(b)] \textit{Quantization Noise}:  
                Let $\boldsymbol{\xi} = \dot{\boldsymbol{x}} - \boldsymbol{x} = [\xi_i]$ denote the quantization noise. Assume $Z \sim \mathcal{U} \left( \left[ -\frac{\Delta}{2}, \frac{\Delta}{2} \right] \right)$ is independent of $Y$. Let $g(u) := \mathbb{E}\left(\exp(\text{i} u Y)) \mathbb{E}(\exp(\text{i} u Z)\right)$. For a given positive integer $p$, if the $p$-th order derivative $g^{(p)}(u)$ satisfies 
                \begin{equation}
                        g^{(p)}\left( \frac{2 \pi l}{\Delta} \right) = 0,
                \end{equation}
                for all non-zero integers $l$. Then the $p$-th conditional moment of $\xi_i$ does not depend on $\boldsymbol{x}$. Specifically,
                \begin{equation}
                        \mathbb{E}[\xi_i^p | \boldsymbol{x}] = \mathbb{E}(Y + Z)^p.
                \end{equation}
        \end{itemize}
\end{lemma}

It can be verified that both uniform dither and triangular dither satisfy the conditions of (\emph{a}) in Lemma~\ref{LemmaOfQuantization} \cite{chen2023quantized, gray1993dithered}.
Moreover, triangular dither also satisfies the condition of (\emph{b}) for $p = 2$.
Under triangular dither, it is further known that
\begin{equation}
        \label{Xi}
        \mathbb{E}[\boldsymbol{\xi}\boldsymbol{\xi}^T] = \dfrac{\Delta^2}{4} \boldsymbol{I}_d,
\end{equation}
which demonstrates that adding appropriate dither before quantization ensures that both the quantization error and quantization noise exhibit statistically desirable properties \cite{chen2023quantizing}.

\section{Ruler-based Quantized Toeplitz Covariance Estimation} \label{sec:rule}

In this section, we formally define the problem of quantized Toeplitz covariance estimation.
Additionally, we propose an unbiased ruler-based estimator tailored for this setting, ensuring accurate estimation under the constraints of quantized data and sparse observation.

\subsection{Problem Description}
We aim to use the quantized data $\dot{\boldsymbol{x}}^{(l)}$ for covariance estimation.
Unlike previous studies on quantized covariance estimation, e.g., \cite{chen2023quantizing}, we focus on the estimation of Toeplitz covariance matrix, and thus use the ruler-based method. Specifically, the problem can be formulated as follows.

\begin{problem}
        \label{Problem_QTCE}
        Let $T \in \mathbb{R}^{d \times d}$ be a symmetric positive semidefinite Toeplitz matrix.
        $\boldsymbol{x}^{(1)}, \ldots, \boldsymbol{x}^{(n)}$ are i.i.d. copies of a zero-mean, $d$-dimensional Gaussian random vector $\boldsymbol{x} \sim \mathcal{N}(0, T)$.
        For a quantization level $\Delta > 0$ and a ruler $R$, assume the availability of quantized observations $\dot{\boldsymbol{x}}^{(l)}_R$, $l = 1, \ldots, n$, where $\dot{\boldsymbol{x}}^{(l)}_R$ denotes the sub-vector of the quantized sample $\dot{\boldsymbol{x}}^{(l)}$ indexed by $R$.
        The goal of quantized Toeplitz covariance estimation (QTCE) is to construct an estimator $\hat{T}$ such that, with probability at least $1-\delta$,
        \begin{equation}
                \|T-\hat{T}\|_2 \leq \epsilon,\qquad \text{or}\qquad \|T-\hat{T}\|_2 \leq \epsilon \|T\|_2,
        \end{equation}
        where $\|\cdot\|_2$ denotes the spectral (operator) norm.
\end{problem}

\begin{figure*}[!t]
    \centering
    \includegraphics[width=0.85\linewidth]{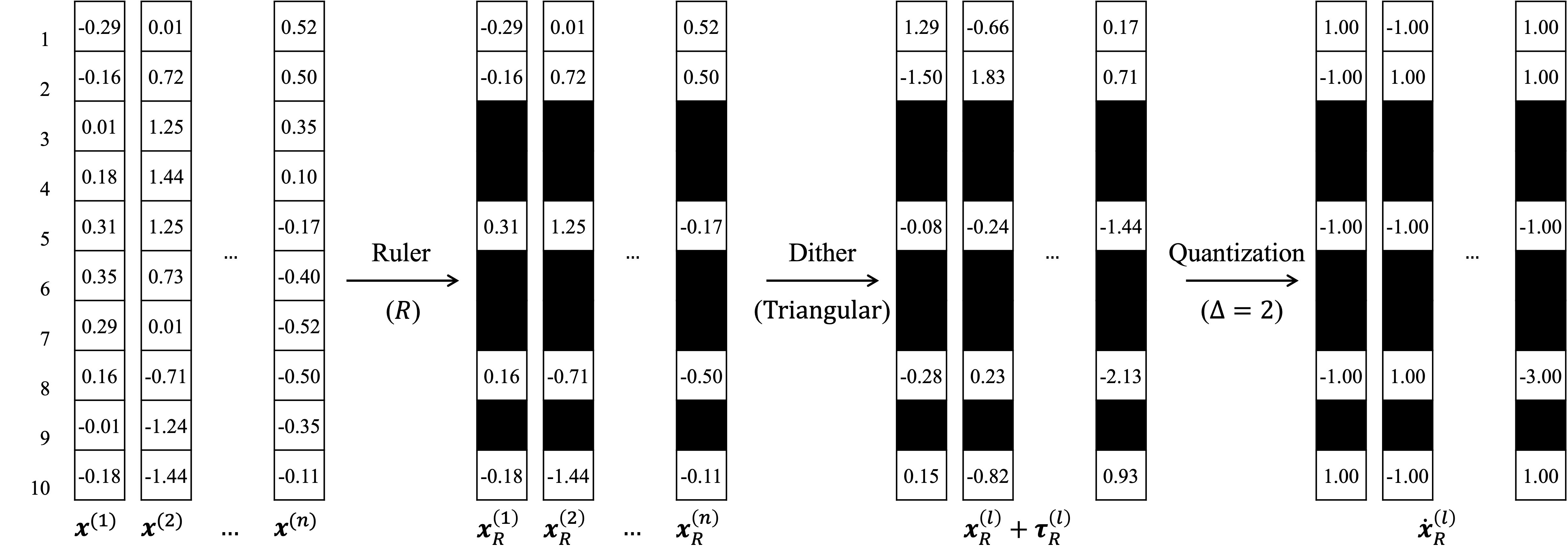}
    \caption{Ruler-based quantized data in the RQTCE framework. In this example, we adopt the sparse ruler $R = \{1, 2, 5, 8, 10\}$ and employ triangular dither for quantization.}
    \label{RQTCE}
\end{figure*}

As illustrated in Fig.~\ref{RQTCE}, the amount of information required for QTCE is minimal.
Compared to the precise samples $\boldsymbol{x}^{(l)}$, we only rely on the quantized sub-vector $\dot{\boldsymbol{x}}_R^{(l)}$ to construct the estimator $\hat{T}$.
To be specific, we neither have access to the precise information of the sample $\boldsymbol{x}_R$ as in traditional Toeplitz covariance estimation \cite{eldar2020sample}, nor do we need to utilize all the information about $\dot{\boldsymbol{x}}$ as in other quantized covariance estimation methods \cite{chen2023quantizing}.
This reduction in required information is due to the inherent Toeplitz structure.

\subsection{Unbiased Estimator}
To address Problem~\ref{Problem_QTCE}, we propose a slight modification to the ruler-based approach.
Specifically, for the given quantization level $\Delta > 0$ and a given ruler $R$, inspired by (\ref{RulerEst}), we define
\begin{equation}
        \label{DotA_L}
        \dot{a}_s^{(l)} = \dfrac{1}{|R_s|} \sum_{(j, k) \in R_s} \dot{x}_j^{(l)} \dot{x}_k^{(l)},
\end{equation}
for $s = 0, \ldots, d-1$.
It represents the quantized approximation for the covariance corresponding to the distance $s$.

In order to obtain an unbiased estimator, we first calculate the mean of $\dot{a}_s^{(l)}$.
Recall that $\dot{x}_j^{(l)} = x_j^{(l)} + \xi_j^{(l)}$, we have
\begin{equation}
        \label{ExpactationDotA_L1}
        \begin{aligned}
                \mathbb{E}(\dot{a}_s^{(l)}) = \dfrac{1}{|R_s|} \sum_{(j, k) \in R_s} \mathbb{E}(\dot{x}_j^{(l)} \dot{x}_k^{(l)}) = \dfrac{1}{|R_s|} \sum_{(j, k) \in R_s} \mathbb{E}\big( (x_j^{(l)} + \xi_j^{(l)})(x_k^{(l)} + \xi_k^{(l)}) \big).
        \end{aligned}
\end{equation}
Note that $T$ is a symmetric Toeplitz matrix, meaning that $T_{j,k}$ is constant for any $(j, k) \in R_s$.
Thus we write
\begin{equation}
        \mathbb{E}(x_j^{(l)} x_k^{(l)}) = T_{j,k} \triangleq a_s,
\end{equation}
where $s = |j - k|$.
Further, under the \textit{triangular dither} defined in (\ref{TriangularDither}), we know
\begin{equation}
        \mathbb{E}(\xi_j^{(l)} \xi_k^{(l)}) = \dfrac{\Delta^2}{4} \delta_{j,k},
\end{equation}
where
\begin{equation}
        \delta_{j,k} = 
        \begin{cases} 
                1, & \text{if } j = k \\
                0, & \text{if } j \neq k
        \end{cases}
\end{equation}
is the Dirac function.
Meanwhile, applying Lemma~\ref{LemmaOfQuantization}, we obtain
\begin{equation}
        \begin{aligned}
                \mathbb{E}(x_j^{(l)} \xi_k^{(l)}) = \mathbb{E}\big(x_j^{(l)} (\omega_k^{(l)} + \tau_k^{(l)})\big)
                = \mathbb{E}(x_j^{(l)} \omega_k^{(l)}) + \mathbb{E}(x_j^{(l)} \tau_k^{(l)}) = 0.
        \end{aligned}
\end{equation}
It is known that $\mathbb{E}(x_k^{(l)} \xi_j^{(l)}) = \mathbb{E}(x_j^{(l)} \xi_k^{(l)}) = 0$ by symmetry.
Thus combining all the above discussions, we conclude that
\begin{equation}
        \label{ExpactationDotA_L2}
        \begin{aligned}
                \mathbb{E}\big( (x_j^{(l)} + \xi_j^{(l)})(x_k^{(l)} + \xi_k^{(l)}) \big)
                = \mathbb{E}(x_j^{(l)} x_k^{(l)}) + \mathbb{E}(x_j^{(l)} \xi_k^{(l)}) + \mathbb{E}(x_k^{(l)} \xi_j^{(l)}) + \mathbb{E}(\xi_j^{(l)} \xi_k^{(l)})
                = a_s + \dfrac{\Delta^2}{4} \delta_s,
        \end{aligned}
\end{equation}
which means $\mathbb{E}(\dot{a}_s^{(l)}) = a_s + \dfrac{\Delta^2}{4} \delta_s$. The result inspires us to define an estimator for $T$ based on $\dot{a}_s^{(l)}$.

\begin{definition}
        \label{EstimatorOfToep}
        Under the settings in Problem~\ref{Problem_QTCE}, the Toeplitz covariance estimator for $T$ is defined as
        \begin{equation}
                \label{Eq_estimator_unbiased_matrix}
                \hat{T} = \text{Toep}(\dot{\boldsymbol{a}}) - \dfrac{\Delta^2}{4} I_d \triangleq \text{Toep}(\hat{\boldsymbol{a}}),
        \end{equation}
        where
        \begin{equation}
                \dot{a}_s = \dfrac{1}{n} \sum_{l=1}^n \dot{a}_s^{(l)} = \dfrac{1}{n |R_s|} \sum_{l = 1}^n \sum_{(j, k) \in R_s} \dot{x}_j^{(l)} \dot{x}_k^{(l)},
        \end{equation}
        is the average value of $\dot{a}_s^{(l)}$.
\end{definition}

In \eqref{Eq_estimator_unbiased_matrix}, the term $\dfrac{\Delta^2}{4} \boldsymbol{I}_d$ arises due to the introduction of the triangular dither.
However, this does not mean that the triangular dither performs poorly.
On the contrary, the use of triangular quantization enables us to calculate the second-order moments of the quantization noise $\boldsymbol{\xi}^{(l)}$ and thus derive (\ref{ExpactationDotA_L2}).
Overall, when triangular dither is employed, $\hat{T}$ serves as an appropriate estimator for $T$.

\begin{theorem}
        Under the settings in Problem~\ref{Problem_QTCE}, if the triangular dither is employed, then
        \begin{equation}
                \label{UnbiasedMean}
                \mathbb{E}(\hat{T}) = T,
        \end{equation}
        which means that the estimator $\hat{T}$ defined in Definition~\ref{EstimatorOfToep} is unbiased.
\end{theorem}
\begin{proof}
        Combining (\ref{ExpactationDotA_L2}) with the fact
        \begin{equation}
                \hat{a}_s = \dot{a}_s - \dfrac{\Delta^2}{4} \delta_s = \dfrac{1}{n} \sum_{l=1}^n \dot{a}_s^{(l)} - \dfrac{\Delta^2}{4} \delta_s,
        \end{equation}
        we can directly get (\ref{UnbiasedMean}).
\end{proof}

Thus far, we have constructed an unbiased estimator of $T$.
However, this construction alone does not fully resolve Problem~\ref{Problem_QTCE}.
To provide a comprehensive solution, it is essential to further analyze the estimation error of $\hat{T}$.
This crucial aspect will be the focus of the subsequent section.

\section{Non-asymptotic Performance Analysis of $\hat{T}$} \label{sec:nona}

In this section, we delve into the analysis of the non-asymptotic error bounds of $\hat{T}$.
While this task presents significant challenges, it is essential for a comprehensive understanding of the estimator's performance.

\subsection{Non-asymptotic Upper Bounds of $\hat{T}$}

We now analyze the error bounds of $\hat{T}$ under both the maximum norm and the operator norm.
The maximum norm, defined as $\|T-\hat{T}\|_{\infty} = \max_{s\in [d]} |a_s - \hat{a}_s|$, captures the entrywise approximation accuracy of $\hat{T}$ relative to $T$.
It is particularly informative for assessing the ability that the estimator recovers individual entries of the covariance matrix.
The corresponding result is stated below.

\begin{theorem}
        \label{Theorem_bound_infty_norm}
        Under the setting of Problem~\ref{Problem_QTCE}, given a fixed quantization level $\Delta > 0$, a prescribed ruler $R$, and assuming the use of triangular dithering, 
        there exists a universal constant $c > 0$ such that
        \begin{equation}
                \label{Eq_bound_infty_norm}
                \begin{aligned}
                        \mathbb{P}\left\{ \| T - \hat{T} \|_{\infty} \geq t \right\} 
                        \leq 2 d \exp{\left[-cn \min \left( \dfrac{t^2}{\mathcal{K}^2}, \dfrac{t}{\mathcal{K}} \right) \right]},
                \end{aligned}
        \end{equation}
        where $\mathcal{K} = 4\|T\|_2 + 8 \Delta^2$. In particular, for any $\delta \in (0, 1]$, if $n \gtrsim \log(d / \delta)$, then with probability at least $1 - \delta$,
        \begin{equation}
                \label{Eq_bound_entrywise}
                \| T - \hat{T} \|_{\infty} \leq C \mathcal{K} \sqrt{\dfrac{\log(d / \delta)}{n}},
        \end{equation}
        where $C > 0$ is a constant that depends only on $c$.
\end{theorem}

\begin{proof}
        See Appendix~\ref{Proof_bound_infty_norm}.
\end{proof}

Theorem~\ref{Theorem_bound_infty_norm} shows that the proposed estimator $\hat{T}$ can achieve arbitrarily high entrywise accuracy in approximating $T$.
This finding offers partial validation of the estimator's effectiveness and serves as a key ingredient for the subsequent analyses.

Subsequently, we shift our focus to the error bounds in terms of the operator norm.
Deriving non-asymptotic bounds for $\hat{T}$ under the operator norm, however, is a challenging task.
While it is known that $\dot{\boldsymbol{x}} = \boldsymbol{x} + \boldsymbol{\tau} + \boldsymbol{\omega}$, the intricate relationship between $\boldsymbol{\omega}$ and $\boldsymbol{\tau}$ leaves the distribution of $\dot{\boldsymbol{x}}$ uncertain.
This uncertainty introduces significant challenges in directly leveraging existing theoretical frameworks.
Fortunately, for Toeplitz matrices, the following fundamental result provides a useful tool for bounding their operator norm \cite{meckes2007spectral}.

\begin{lemma}
        \label{Lemma_toeplitz_operator_norm_symbol}
        Let $\boldsymbol{e} = (e_0, \ldots, e_{d-1})^\top \in \mathbb{R}^d$.
        The operator norm of the Toeplitz matrix $\mathrm{Toep}(\boldsymbol{e})$ is bounded by
        \begin{equation}
                \| \mathrm{Toep}(\boldsymbol{e}) \|_2 \leq \sup_{x \in [0,1]} L_{\boldsymbol{e}}(x),
        \end{equation}
        where the associated \textit{spectral density function} $L_{\boldsymbol{e}} : [0,1] \to \mathbb{R}$ is defined as
        \begin{equation}
                L_{\boldsymbol{e}}(x) = e_0 + 2 \sum_{s=1}^{d-1} e_s \cos(2\pi s x).
        \end{equation}
\end{lemma}

Lemma~\ref{Lemma_toeplitz_operator_norm_symbol} establishes a bound on the operator norm of a Toeplitz matrix via the spectral density function $L_{\boldsymbol{e}}(x)$.
This formulation reduces the task of bounding $T - \hat{T}$ to analyzing the supremum of $L_{\boldsymbol{e}}(x)$.
Building on this lemma, we obtain the following theorem.

\begin{theorem}
        \label{Theorem_bound_operator_norm}
        Under the setting of Problem~\ref{Problem_QTCE}, given a fixed quantization level $\Delta > 0$, a prescribed ruler $R$, and assuming the use of triangular dithering,
        there exists a universal constant $c > 0$ such that
        \begin{equation}
                \label{Eq_bound_spectral_norm_t}
                \mathbb{P}\left\{\|T - \hat{T}\|_2 \geq t \right\} \leq 32 \pi d^2 \exp{\left(-cn \kappa \right)},
        \end{equation}
        where $\mathcal{K} = 4 \|T\|_2 + 8\Delta^2$ and $\kappa = \min \left( \dfrac{t^2}{\mathcal{K}^2 \phi(R)}, \dfrac{t}{\mathcal{K} \sqrt{\phi(R)}} \right)$. 
        In particular, for any $\delta \in (0,1]$, if $n \gtrsim \log(d / \delta)$, then with probability at least $1 - \delta$,
        \begin{equation}
                \label{Eq_bound_operator_t}
                \|T - \hat{T}\|_2 \leq C \mathcal{K} \sqrt{\dfrac{\phi(R) \cdot \log(d / \delta)}{n}},
        \end{equation}
        where $C > 0$ is a constant that depends only on $c$.
\end{theorem}

\begin{proof}
        In this proof, we relate the operator norm error $\|T - \hat{T}\|_2$ to its associated spectral density function $L_{\boldsymbol{a} - \hat{\boldsymbol{a}}}(x)$.
        Specifically, we complete the proof by reducing the operator norm error $\|T - \hat{T}\|_2$ to a uniform bound on $L_{\boldsymbol{a} - \hat{\boldsymbol{a}}}(x)$ over $x \in [0,1]$.
        That is, we establish
        \begin{equation}
                \label{Eq_bound_spectral_density}
                \mathbb{P}\left\{\sup_{x\in [0,1]} |L_{\boldsymbol{e}} (x)| \leq t \right\} \geq 1 - 32 \pi d^2 \exp{\left(-cn \kappa \right)},
        \end{equation}
        This is a standard and effective technique that has been widely adopted in the literature \cite{han2012covariance, eldar2020sample, klockmann2024efficient}.
        At this point, the focus of the proof shifts to establishing \eqref{Eq_bound_spectral_density}, which consists of the following three steps:
        Firstly, we analyze the element-wise error $a_s - \hat{a}_s$, which has already been established in Theorem~\ref{Theorem_bound_infty_norm}. This step quantifies the influence of quantization on each individual coefficient.
        Secondly, we establish a pointwise bound on the spectral density function $L_{\boldsymbol{a} - \hat{\boldsymbol{a}}}(x)$ for a fixed $x \in [0,1]$. This step is technically more involved, as it requires carefully handling missing entries induced by the sampling pattern specified by the ruler $R$.
        Finally, we convert the uniform bound \eqref{Eq_bound_spectral_density} by leveraging a covering argument, together with the previous two steps.
        Full technical details are provided in Appendix~\ref{Proof_bound_operator_norm}.
\end{proof}

So far our results have been derived under the assumption that $\boldsymbol{x}^{(1)}, \ldots, \boldsymbol{x}^{(n)}$ are Gaussian.
We note that in the proofs, we mainly rely on their sub-Gaussian nature, meaning that the results can be readily generalized to the sub-Gaussian case with a minor adjustment to \eqref{Eq_sub_gaussian_norm_of_x}.
This generalization ensures that our results offer a more general theoretical guarantee for $\hat{T}$. Furthermore, our proposed estimator $\hat{T}$ demonstrates the capability to approximate $T$ with arbitrary precision $\epsilon$ in terms of the operator norm. This result highlights the potential of performing covariance estimation using low-resolution data, making it particularly valuable in scenarios where data resolution is constrained by practical limitations.

\begin{remark}
        In Theorem~\ref{Theorem_bound_operator_norm}, we also showed that uniformly consistent spectral density estimation of a Gaussian time series can be achieved by simply setting \(\hat{L}_{\boldsymbol{a}}(x) = L_{\hat{\boldsymbol{a}}}(x)\).
        Specifically, it holds that
        \begin{equation}
                \mathbb{P}\left\{\sup_{x\in [0,1]} |L_{\boldsymbol{a}} (x) - \hat{L}_{\boldsymbol{a}} (x)| \leq t \right\} \geq 1 - 32 \pi d^2 \exp{\left(-cn \kappa \right)},
        \end{equation}
        which corresponds to \eqref{Eq_bound_spectral_density}.
        This result provides a non-asymptotic bound for spectral density estimation, which is of independent interest in spectral analysis \cite{stoica2005spectral, loubaton2023asymptotic, zhang2025spectral}.
\end{remark}

\subsection{Upper Bound under the Specific Ruler $R_\alpha$}
The upper bound in Theorem~\ref{Theorem_bound_operator_norm} applies to general ruler $R$, but its value is affected by the coverage coefficient $\phi(R)$.
As noted earlier, $\phi(R)$ depends on multiple factors and lacks a simple closed-form expression.
To offer a more concrete illustration of Theorem~\ref{Theorem_bound_operator_norm}, we focus in this part on the specific ruler $R_\alpha$ defined in Definition~\ref{SpecialRuler}.

While it is challenging to analyze the value of $\phi(R)$ for a general ruler, for the ruler $R_{\alpha}$ defined in Definition~\ref{SpecialRuler}, we have the following useful lemma.

\begin{lemma}
        \label{Lemma_bound_phi_r_alpha}
        For any $\alpha \in [1/2,1]$, the cardinality of the ruler $R_{\alpha}$ satisfies $|R_{\alpha}| = O\left(d^{\alpha}\right)$, and the corresponding coverage coefficient admits the bound
        \begin{equation}
                \label{Eq_ruler_phi_r_alpha}
                \phi(R_\alpha) 
                \leq 2 d^{2 - 2\alpha} + d^{1 - \alpha} \left(1 + \log \left(\left\lceil d^{2 \alpha - 1} \right\rceil\right)\right)
                = d^{2 - 2\alpha} + O(d^{1 - \alpha} \cdot \log d).
        \end{equation}
\end{lemma}

\begin{proof}
        See Appendix~\ref{Proof_bound_phi_r_alpha}.
\end{proof}

Note that Lemma~\ref{Lemma_bound_phi_r_alpha} is consistent with (\ref{RRRRRuler}) and provides an upper bound for $\phi(R_\alpha)$.
For a general ruler $R$, once an upper bound for $\phi(R)$ is available, the same line of analysis applies.
The tightness of this bound directly impacts the sharpness of the error bound in \eqref{Eq_bound_operator_t}.
Combining Lemma~\ref{Lemma_bound_phi_r_alpha} with Theorem~\ref{Theorem_bound_operator_norm} immediately yields the following corollary.

\begin{corollary}
        \label{Corollary_upper_bound_r_alpha}
        For any $\alpha \in [1/2, 1]$, consider the use of a specialized ruler $R_\alpha$.
        Under the conditions in Theorem~\ref{Theorem_bound_operator_norm}, for any $\delta \in (0, 1]$, if $n \gtrsim \log(d/\delta)$, then with probability at least $1-\delta$,
        \begin{equation}
                \|T - \hat{T}\|_2 \leq C \mathcal{K} \sqrt{\dfrac{\max(d^{2-2\alpha}, d^{1-\alpha} \log d) \cdot \log(d / \delta)}{n}},
        \end{equation}
        where $\mathcal{K} = 4 \|T\|_2 + 8\Delta^2$ and $C>0$ is a constant. In particular, for any $\epsilon > 0$, if
        \begin{equation}
                \label{Eq_r_alpha_upper_bound}
                n \gtrsim_{\mathcal{K}} \dfrac{\max(d^{2-2\alpha}, d^{1-\alpha} \log d) \cdot \log(d / \delta)}{\epsilon^2}
        \end{equation}
        then $\|T - \hat{T}\|_2 \leq \epsilon$.
\end{corollary}

Notably, when $\Delta = 0$, we have $\mathcal{K} = 4 \|T\|_2$, and Corollary~\ref{Corollary_upper_bound_r_alpha} then reduces to \cite[Theorem 4.4]{eldar2020sample}, aside from minor differences in the constants.
Remarkably, when $\Delta \neq 0$, the bound in Corollary~\ref{Corollary_upper_bound_r_alpha} degrades only by a factor of $\mathcal{K}$, indicating that even with coarse quantized data, the convergence rate remains of the same order as in the unquantized case.

\begin{remark}
        From Corollary~\ref{Corollary_upper_bound_r_alpha}, we observe the inherent trade-off between VSC, ESC, and resolution in Toeplitz covariance estimation.
        Specifically, VSC corresponds to the number of samples $n$, while ESC is related to the cardinality of the ruler $R$, which is closely tied to the coverage coefficient $\phi(R)$.
        The resolution, on the other hand, is determined directly by the quantization level $\Delta$.
        A larger $\Delta$ implies a lower resolution.
        Corollary~\ref{Corollary_upper_bound_r_alpha} demonstrates that reducing one or more of these quantities (VSC, ESC or resolution) necessitates compensatory increases in the remaining quantities to maintain the desired estimation accuracy.
        For instance, for a fixed ruler $R_\alpha$, reducing the resolution of the data (i.e., increasing $\Delta$) requires an increase in the number of samples (i.e., VSC) to achieve the same level of accuracy.
        This interplay highlights the necessity of balancing these factors in practical covariance estimation tasks.
\end{remark}

To further illustrate the implications of Corollary~\ref{Corollary_upper_bound_r_alpha}, we examine two special cases: the \emph{full ruler} $R_1$ and the \emph{sparse ruler} $R_{1/2}$.
These cases provide contrasting perspectives on the trade-offs between VSC, ESC and resolution.
The results are summarized in the following corollary.

\begin{corollary}
        \label{Corollary_1_12}
        Under the conditions of Corollary~\ref{Corollary_upper_bound_r_alpha}, it holds that $\|T-\hat{T}\|_2 \leq \epsilon$ with probability at least $1 - \delta$ if any of the two conditions holds:
    \begin{enumerate}
        \item The full ruler $R_1$ is used with the VSC
        \begin{equation}
            n_1 \gtrsim_{\mathcal{K}} \dfrac{\log d \cdot \log(d / \delta)}{\epsilon^2},
        \end{equation}
        
        \item or, the sparse ruler $R_{1/2}$ is used with the VSC
        \begin{equation}
                n_{1/2} \gtrsim_{\mathcal{K}} \dfrac{d \cdot \log(d / \delta)}{\epsilon^2}
        \end{equation}
    \end{enumerate}
    where $\mathcal{K} = 4 \|T\|_2 + 8\Delta^2$.
\end{corollary}

Corollary~\ref{Corollary_1_12} demonstrates that for a prescribed tolerance $\epsilon$ and quantization level $\Delta$, reducing the ESC from $|R_1| = d$ to $|R_{1/2}| = O(\sqrt{d})$ inevitably increases VSC.
Despite this trade-off, the performance remains highly competitive: even with the sparse ruler $R_{1/2}$, the proposed estimator $\hat{T}$ attains a convergence rate of $O\left(\sqrt{\dfrac{d \log{(d / \delta)}}{n}}\right)$, which is comparable to the result in \cite[Theorem 3]{chen2023quantizing} where the full samples are observed.
Importantly, quantization only affects the coefficients of the convergence rate, leaving the convergence order unchanged.
These results suggest that the proposed estimator $\hat{T}$ is near optimal.

\subsection{Non-asymptotic Lower Bound of $\hat{T}$}
In the preceding sections, we established an upper bound on the estimation error.
To complement this result and verify its tightness, we now turn to the analysis of a corresponding lower bound.
In particular, for the specific ruler $R_\alpha$, we investigate the minimum sample size required to attain a prescribed accuracy in a statistical sense, stated as follows.

\begin{theorem}
        \label{Theorem_lower_bound}
        Let $R \subset [d]$ be a sparse ruler with $|R| \leq d^\alpha$ for some $\alpha \in [1/2, 1]$.
        Under the conditions in Theorem~\ref{Theorem_bound_operator_norm}, let $\epsilon > 0$ be sufficiently small.
        Consider any (possibly randomized) algorithm that only has access to quantized observations $\dot{\boldsymbol{x}}^{l}_R$, $l=1,2, \ldots, n$ before outputting $\tilde{T}$.
        If $\|T - \tilde{T}\|_2 \leq \epsilon$ holds for all PSD Toeplitz covariance matrices $T$ with probability at least $1/10$, then it must hold that
        \begin{equation}
                \label{Eq_lo_b}
                n \gtrsim \dfrac{\max(d^{3-4\alpha}, d^{1-\alpha})}{\epsilon^2}.
        \end{equation}
\end{theorem}

\begin{proof}
        The proof is completed by applying Assouad's Lemma \cite{assouad1983deux, eldar2020sample} and Le Cam's method \cite{yu1997assouad}.
        Full details are presented in Appendix~\ref{Proof_lower_bound}.
\end{proof}

When $\alpha = 1$ or $1/2$, the conclusion in Theorem~\ref{Theorem_lower_bound} is consistent with \eqref{Eq_r_alpha_upper_bound}, up to logarithmic factors.
In the intermediate regime $\alpha\in(1/2,1)$, the bounds in \eqref{Eq_lo_b} and \eqref{Eq_r_alpha_upper_bound} do not coincide; nevertheless, they remain informative and delineate the qualitative behavior of the problem.
Closing this gap and establishing matching bounds in this regime is an interesting open problem for future work.
It worth noting that Theorem~\ref{Theorem_lower_bound} yields a tighter bound than \cite[Theorem~4.5]{eldar2020sample} in the non-quantized Gaussian setting.

\subsection{Positive Semi-definiteness of $\hat{T}$}
Obtaining a positive semidefinite (PSD) estimator is often essential for subsequent algorithms that rely on the estimated covariance matrix.
However, for Toeplitz-structured estimators, ensuring PSDness is particularly challenging.
From a theoretical perspective, the proposed estimator $\hat{T}$ is not guaranteed to be PSD in all cases.
Nevertheless, the operator-norm error bound established in Theorem~\ref{Theorem_bound_operator_norm} implies that $\hat{T}$ becomes positive definite with high probability when the sample size $n$ is sufficiently large.
Specifically, we establish the following result.

\begin{theorem}
        \label{Theorem_psd}
        Under the conditions in Theorem~\ref{Theorem_bound_operator_norm}, suppose $T$ is positive definite (PD) with the minimum eigenvalue $\lambda_d(T) > 0$. For any $\delta \in (0, 1]$, if
        \begin{equation}
                \label{Eq_psd_n}
                n > \dfrac{c_0 C^2 \mathcal{K}^2 \phi(R) \cdot \log (d / \delta)}{\lambda_d^2(T)}
        \end{equation}
        for some constant $c_0 > 1$, then $\hat{T}$ is positive definite with probability at least $1 - \delta$.
\end{theorem}

While the result is conservative, empirical evidence from simulation experiments suggests that the estimated covariance matrix is typically positive semidefinite in practice.

Widely known as the precision matrix, the inverse covariance matrix is essential to tasks such as classification, estimation, and risk modeling across domains like machine learning and finance.
Accurate estimation of the precision matrix is essential, as errors in its approximation can propagate and significantly affect subsequent tasks.
Motivated by these considerations, we further derive error bounds for the precision matrix.

\begin{theorem}
        \label{Theorem_precision}
        Under the conditions in Theorem~\ref{Theorem_psd}, the PD estimator $\hat{T}$ satisfies
        \begin{equation}
                \label{Eq_psd_bound_inverse_t}
                \|T^{-1}-\hat{T}^{-1}\|_2
                \leq \dfrac{C_{inv} \mathcal{K}}{\lambda_d^2(T)} \sqrt{\dfrac{\phi(R) \cdot \log (d / \delta)}{n}},
        \end{equation}
        where $C_{inv}$ is a constant depending only on $C$ and $c_0$.
\end{theorem}

\begin{proof}
        Both this result and Theorem~\ref{Theorem_psd} follow directly from Theorem~\ref{Theorem_bound_operator_norm} by applying Weyl's inequality.
        Detailed derivations are provided in Appendix~\ref{Proof_psd}.
\end{proof}

This result indicates that the precision matrix satisfies the same bounds as the covariance matrix, differing only in the constant factors.
Consequently, with high probability, our proposed method is capable of estimating the precision matrix with high accuracy.

\section{Finite-Bit Estimator and Non-Asymptotic Performance Guarantees} \label{sec:fini}
In the previous section, we introduced an unbiased estimator $\hat{T}$ and conducted a detailed analysis of its error bounds.
A limitation of that approach is that, although the parameter $\Delta$ accounts for the effect of data resolution on approximation accuracy, the quantizer under consideration allows an infinite number of potential outputs and is thus not directly applicable to finite-bit sampling scenarios.
In practice, however, many applications are constrained to only a few bits per entry.
In this section, we address this limitation by constructing a finite-bit estimator through a careful selection of $\Delta$.

\subsection{The $k$-bit Estimator and Its Error Bound}
The quantized data in the previous section is infinite-bit because the underlying Gaussian random variable is unbounded.
However, for Gaussian variables, most of their probability mass is concentrated around the expectation, allowing the tail probabilities to be tightly controlled.
By applying appropriate truncation, we can transform them into bounded random variables, whose quantized outputs become finite-bit.
For instance, \cite{chen2025parameter} achieves a 2-bit quantizer through suitable truncation.
In general, obtaining $k$-bit quantized data requires that the quantizer output contain at most $2^k$ distinct levels.
Following this principle, we consider the $k$-bit quantizer defined as
\begin{equation}
        \label{Eq_quantiter_k_bit}
        \mathcal{Q}_{\Delta, k} (x) :=
        \begin{cases}
                (2^{k-1} + 1/2) \Delta, &\quad \mathrm{if}\  x \geq (2^{k-1} - 1) \Delta,\\
                -(2^{k-1} + 1/2) \Delta, &\quad \mathrm{if}\  x < (1 - 2^{k-1}) \Delta,\\
                \mathcal{Q}_{\Delta}(x), &\quad \mathrm{otherwise}.
        \end{cases}
\end{equation}

We denote by $\hat{T}_k$ the estimator constructed using the $k$-bit quantizer\footnote{Both $\hat{T}_k$ and $\hat{T}$ are constructed in exactly the same way, except that in the construction of $\hat{T}_k$, each sample entry $x_j^{(l)}$ is quantized using $\mathcal{Q}_{\Delta, k}(\cdot)$ instead of $\mathcal{Q}_{\Delta}(\cdot)$.
Here, $k$ bits means that each quantized sample entry $\dot{x}_j^{(l)}$ is represented using only $k$ bits.}.
For Gaussian data, if $\Delta$ is chosen sufficiently large, then $\mathcal{Q}_{\Delta, k}(\cdot) = \mathcal{Q}_{\Delta}(\cdot)$ holds with high probability.
In this case, $\hat{T}_k$ inherits the same error bound as $\hat{T}$ with high probability.
Formally, we have the following theorem.

\begin{theorem}
        \label{Theorem_finite_bit_quant}
        For any prescribed integer $k \geq 2$ and any $\delta, \delta'\in(0,1]$, consider the Toeplitz estimator $\hat{T}_k$ constructed with the $k$-bit quantizer defined in \eqref{Eq_quantiter_k_bit}.
        Under the setting of Problem~\ref{Problem_QTCE}, given a prescribed ruler $R$, and assuming the use of triangular dithering,
        if we set
        \begin{equation}
                \label{Eq_finite_bit_delta}
                \Delta = C_{\mathrm{bit}} \cdot 2^{-k} \sqrt{ \|T\|_{\infty} \log\left( \frac{2n|R|}{\delta'} \right)}
        \end{equation}
        with $C_{bit} \geq \sqrt{2}$, and choose $n \gtrsim \log (d / \delta)$, 
        then there exists a constant $C>0$ such that, with probability at least $1 - \delta - \delta'$, the $k$-bit estimator $\hat{T}_k$ satisfies
        \begin{equation}
                \label{Eq_bound_finite_bit_t_hat_k}
                \begin{aligned}
                        \|T - \hat{T}_k\|_2 \leq C \mathcal{K} \sqrt{\dfrac{\phi(R) \cdot \log(d / \delta)}{n}}
                \end{aligned}
        \end{equation}
        where $\mathcal{K} = 4\|T\|_2 + 8 \Delta^2$.
\end{theorem}

\begin{proof}
        In this proof, we first establish that $\hat{T}_k = \hat{T}$ holds with overwhelming probability.
        Conditional on this event, it is immediate that $\hat{T}_k$ and $\hat{T}$ share the same error bounds with high probability.
        Detailed arguments are provided in Appendix~\ref{Proof_finite_bit_quant}.
\end{proof}

In Theorem~\ref{Theorem_finite_bit_quant}, choosing an excessively large value of $C_{bit}$ reduces the data resolution, whereas choosing it too small leads to excessive truncation of tail data.
Therefore, an appropriate value of $C_{bit}$ should be selected in practice to balance these effects.

\begin{remark}
        When the full ruler (i.e., $\alpha = 1$) is employed, \eqref{Eq_bound_finite_bit_t_hat_k} simplifies to
        \begin{equation}
                \|T - \hat{T}_k\|_2 \lesssim (\|T\|_2 + 8 \Delta^2) \sqrt{\dfrac{\log d \cdot \log(d / \delta)}{n}}
        \end{equation}
        where $\Delta$ is as defined in \eqref{Eq_finite_bit_delta}.
        In particular, for the 2-bit case, we have
        \begin{equation}
                \|T - \hat{T}_k\|_2 \lesssim \sqrt{\dfrac{\|T\|_\infty^2 \log d \cdot \log(d / \delta) \cdot \log^2(2nd / \delta')}{n}},
        \end{equation}
        which is tighter than the general 2-bit covariance estimation bound given in \cite[Theorem 4]{chen2025parameter}.
        This improvement stems from exploiting the Toeplitz structure.
\end{remark}

Building on Theorem~\ref{Theorem_finite_bit_quant}, we derive the following result.

\begin{corollary}
        \label{Corollary_finite_bit}
        Under the conditions in Theorem~\ref{Theorem_finite_bit_quant}, for any $p, p' > 2$,
        suppose we set
        \begin{equation}
                \Delta = C_{\mathrm{bit}} \cdot 2^{-k} \sqrt{ \|T\|_{\infty} \log\left( \frac{2n|R|}{\delta'} \right)}
        \end{equation}
        with $C_{bit} \geq \sqrt{2}$, and choose
        \begin{equation}
                \label{Eq_n_trade_off_k}
                n \gtrsim \max\left\{ p \log d, \dfrac{1}{|R|} \exp \left( \dfrac{2^{2k} \|T\|_2}{p' \|T\|_{\infty}}\right)\right\}.
        \end{equation}
        Then it holds that
        \begin{equation}
                \| T - \hat{T}_k \|_2 \lesssim \dfrac{p' \|T\|_{\infty} \log (n |R|)}{2^{2k}} \cdot \sqrt{\dfrac{p \phi(R) \log(d)}{n}}
        \end{equation}
        with probability at least $1 - d^{1 - p} - 2 (n |R|)^{1 - p'}$.
\end{corollary}

\begin{proof}
        This is an immediate consequence of Theorem~\ref{Theorem_finite_bit_quant} with $\delta = d^{1-p}$ and $\delta' = 2(n|R|)^{1-p'}$.
\end{proof}

It should be noted that the second term on the right-hand side of \eqref{Eq_n_trade_off_k} is to ensure that $\Delta^2 \gtrsim \|T\|_2$, which is naturally satisfied when $k$ is small.
In fact, if this condition is violated, we have
\begin{equation}
        \| T - \hat{T}_k \|_2 \lesssim \|T\|_2 \cdot \sqrt{\dfrac{p \phi(R) \log(d)}{n}}
\end{equation}
In this scenario, the upper bound on the error becomes independent of $k$, implying that increasing the data resolution yields a notable reduction in estimation error only when the sample size is sufficiently large or $k$ is relatively small.
This phenomenon will be further validated in our subsequent experiments.

\subsection{Detailed Analysis of the Trade-off among VSC, ESC, and Resolution}
Building on the error bound for the $k$-bit estimator established in Theorem~\ref{Theorem_finite_bit_quant}, we next examine the trade-off among VSC, ESC, and resolution (i.e., $k$), using $R_\alpha$ as a representative example.
This leads to the following corollary.

\begin{corollary}
        \label{Corollary_ruler_alpha}
        For any \(\alpha \in [1/2, 1]\), consider the use of a specialized ruler $R_{\alpha}$.
        Under the conditions in Theorem~\ref{Theorem_finite_bit_quant}, for any $p, p' > 2$,
        suppose we set
        \begin{equation}
                \Delta = C_{\mathrm{bit}} \cdot 2^{-k} \sqrt{ \|T\|_{\infty} \log\left( \frac{2n|R_\alpha|}{\delta'} \right)}
        \end{equation}
        with $C_{bit} \geq \sqrt{2}$, and choose $n \gtrsim \max\left\{ p \log d, \dfrac{1}{|R_\alpha|} \exp \left( \dfrac{2^{2k} \|T\|_2}{p' \|T\|_{\infty}}\right)\right\}$.
        Then, for any prescribed tolerance $\epsilon > 0$, in order to guarantee
        \begin{equation}
                \|T - \hat{T}_k\|_2 \leq \epsilon,
        \end{equation}
        with probability at least $1 - d^{1 - p} - 2 (n |R|)^{1 - p'}$,
        the following trade-off between the vector sample complexity (VSC) $n$, the entry sample complexity (ESC) $|R_{\alpha}|$, and the number of quantization bits per entry $k$ must be satisfied
        \begin{equation}
                \label{Eq_trade_off_r_alpha}
                p' \|T\|_{\infty} \log (n|R_\alpha|) \cdot 
                \sqrt{\dfrac{p \max\left\{d^2 / |R_\alpha|^2, d / |R_\alpha| \cdot \log(d)\right\} \cdot \log(d)}{n}} 
                \lesssim \epsilon \cdot 2^{2k}.
        \end{equation}
\end{corollary}

\begin{proof}
        This result follows directly by substituting the bound in Lemma~\ref{Lemma_bound_phi_r_alpha} into Corollary~\ref{Corollary_finite_bit}.
\end{proof}

Corollary~\ref{Corollary_ruler_alpha} provides an explicit complexity trade-off.
Building on this result, we refine the relationship between the key parameters.
By squaring both sides of~\eqref{Eq_trade_off_r_alpha}, we obtain
\begin{equation}
        n 
        \gtrsim_{p, p', \|T\|_{\infty}} \dfrac{\log^2 (n|R_\alpha|) \log(d)}{\epsilon^2 \cdot 2^{4k}} 
        \cdot \max\left\{d^2 / |R_\alpha|^2, d / |R_\alpha| \cdot \log(d)\right\}
\end{equation}
To highlight the scaling behavior, we absorb logarithmic factors into the \(\widetilde{O}\)-notation, yielding
\begin{equation}
        \label{Eq_trade_off_r_alpha_simplify}
        n =
        \widetilde{O} \left( \dfrac{\log(d)}{\epsilon^2 \cdot 2^{4k}} 
        \cdot \max\left\{\dfrac{d^2}{|R_\alpha|^2}, \dfrac{d \log(d)}{|R_\alpha|}\right\}\right).
\end{equation}
The constraint in \eqref{Eq_trade_off_r_alpha_simplify} explicitly characterizes the interplay among VSC, ESC, and quantization resolution.
In particular, fixing any one of these parameters induces explicit trade-offs between the remaining two.
Taking into account that $|R_\alpha| = O(d^\alpha)$, we summarize the resulting trade-off as follows.

For a fixed quantization bits $k$, there exists an order-of-magnitude relationship between the vector sample complexity $n$ and the entry sample complexity parameter \(\alpha\), i.e.,
\begin{equation}
        n =
        \widetilde{O} \left( \dfrac{\log(d)}{\epsilon^2 \cdot 2^{4k}} 
        \cdot \max\left\{d^{2 - 2\alpha}, d^{1 - \alpha} \log(d)\right\}\right).
\end{equation}
As $\alpha$ approaches $1/2$, the VSC scales linearly with $d$, whereas for $\alpha \rightarrow 1$, it approaches a $\log^2(d)$-type growth.
This reveals a fundamental trade-off: denser sampling patterns reduce the required number of samples at the cost of observing more entries per sample.

When the entry sample complexity $|R_\alpha|$ is fixed, the relationship between $n$ and $k$ exhibits an exponential dependence.
Specifically,
\begin{equation}
        n =
        \widetilde{O} \left( \dfrac{C(d, \alpha)}{\epsilon^2 \cdot 2^{4k}} \right),
\end{equation}
where $C(d, \alpha)$ is a constant that depends only on $d$ and $\alpha$.
This result implies that for low-bit quantization, more samples are required to compensate for the information loss due to coarse quantization.
Conversely, increasing the number of quantization bits leads to an exponential reduction in the required sample size\footnote{This exponential decay in required sample size $n$ with increasing $k$ holds primarily when $k$ is small. As $k$ becomes large (e.g., from $k=6$ to $k=8$), the marginal gain diminishes because the quantized observations are already sufficiently informative. We will provide further empirical evidence for this phenomenon in the subsequent experiments.}.
This quantifies the inherent balance between VSC and quantization precision.

With the vector sample complexity $n$ fixed, the relationship between $k$ and $|R_\alpha|$ is characterized by
\begin{equation}
        \label{Eq_trade_off_fixed_n}
        \log^2(n|R_\alpha|) \cdot \max\left\{ \frac{d^2}{|R_\alpha|^2},\ \frac{d\log(d)}{|R_\alpha|} \right\} \lesssim_{p,p',\|T\|_\infty} \frac{n \cdot \epsilon^2 \cdot 2^{4k}}{ \log(d) }.
\end{equation}
This highlights the trade-off between quantization accuracy and observation sparsity (i.e., ESC):
increasing the number of quantization bits \(k\) exponentially enlarges the right-hand side of~\eqref{Eq_trade_off_fixed_n}, thus allowing for sparser measurements.

\section{Low-rank and Banded Cases} \label{sec:lowr}

For Toeplitz matrices, the low-rank and banded cases are of particular interest due to their prevalence in practical applications \cite{maly2022new}.
In this section, we focus on these two special cases by incorporating the low-rank and band-limit properties into the quantized Toeplitz covariance estimation problem.
By leveraging these structural assumptions, we aim to achieve a faster rate of convergence compared to the general case.

\subsection{Low-rank Case}
Low-rankness is a property often considered in matrix analysis \cite{negahban2011estimation}.
Therefore, we further explore the effect of low-rankness on the rate of convergence.
However, upon examining the conclusion in Theorem \ref{Theorem_bound_operator_norm}, it appears that low-rankness cannot be directly related to the results presented there.
So we have to make a slight modification to it.

For low-rank Toeplitz matrices, we rely on the following lemma from \cite{eldar2020sample}:

\begin{lemma}
        \label{Lemma_low_rank_bound}
        Let $\alpha \in [1/2, 1]$. Given a positive semidefinite Toeplitz matrix $T \in \mathbb{R}^{d \times d}$ and the sparse ruler $R_{\alpha}$, 
        for all $r \leq d$ the following inequality holds
        \begin{equation}
                \|T_{R_{\alpha}}\|_2^2
                \leq \dfrac{32 r^2}{d^{2-2\alpha}} \cdot \|T\|_2^2 + 8 \lambda(r, T),
        \end{equation}
        where
        \begin{equation}
                \lambda(r, T) = \min \left\{ \|T - T_{r}\|_2^2,\ \dfrac{2 \|T - T_{r}\|_F^2}{d^{1-\alpha}} \right\}.
        \end{equation}
        Here, $T_{r}$ denotes the best rank-$r$ approximation of $T$, i.e.,
        \begin{equation}
                T_{r} = \arg \min_{\mathrm{rank}(M)=r} \|T - M\|_F = \arg \min_{\mathrm{rank}(M)=r} \|T - M\|_2.
        \end{equation}
\end{lemma}

Lemma~\ref{Lemma_low_rank_bound} shows that for a rank-$r$ Toeplitz matrix $T$, the operator norm of its principal submatrix $T_{R_{\alpha}}$ can be bounded by $\|T\|_2$, with the bound explicitly depending on $r$.
Inspired by Lemma~\ref{Lemma_low_rank_bound}, a slight modification to the proof procedure of Theorem~\ref{Theorem_bound_operator_norm} leads to the following result.

\begin{theorem}
        \label{Theorem_low_rank}
        Under the setting of Problem~\ref{Problem_QTCE}, consider the low-rank scenario where $\mathrm{rank}(T) = r$.
        For any $\delta \in (0, 1]$ and $\alpha \in [1/2, 1]$, given a fixed quantization level $\Delta > 0$, a prescribed ruler $R_\alpha$, and assuming the use of triangular dithering,
        if $n \gtrsim \log(d / \delta)$, then there exists a universal constant $C > 0$ such that, with probability at least $1 - \delta$,
        \begin{equation}
                \label{Eq_low_rank}
                \| T - \hat{T} \|_2 \leq C \mathcal{K}_r \sqrt{\dfrac{\log(d / \delta) \cdot \max\left\{ d^{2 - 2\alpha}, d^{1-\alpha} \log(d) \right\}}{n}},
        \end{equation}
        where $\mathcal{K}_r = \dfrac{24 r}{d^{1-\alpha}} \|T\|_2 + 8 \Delta^2$.
\end{theorem}

\begin{proof}
        This proof closely follows the derivation of Theorem~\ref{Theorem_bound_operator_norm}, with two key refinements:
        we tighten the sub-Gaussian norm bound of $\dot{x}_j^{(l)}$ and incorporate Lemma~\ref{Lemma_low_rank_bound} to leverage the low-rank structure.
        Detailed arguments are provided in Appendix~\ref{Proof_low_rank}.
\end{proof}

It follows from Theorem~\ref{Theorem_low_rank} that the condition $6 r < d^{1 - \alpha}$ allows the low-rank property to be more effectively exploited when a sparser ruler is employed.
This indicates that, in certain scenarios, low-rankness can substantially improve the convergence rate.
When $6 r < d^{1 - \alpha}$ holds, the low-rank structure yields an accelerated convergence rate, providing notable practical benefits.
Moreover, Theorem~\ref{Theorem_low_rank} also shows that low-rankness mainly affects the constant factor in the convergence rate (i.e., $\mathcal{K}_r$) rather than its asymptotic order.

\subsection{Banded Case}
We now turn our attention to the banded case, a fundamental property that has been extensively studied in the context of Toeplitz covariance estimation \cite{gray2006toeplitz}.
In this setting, we assume that $T$ is banded, meaning that only a limited number of diagonals near the main diagonal contain non-zero entries.
Specifically, for the Toeplitz matrix $T = \text{Toep}(\boldsymbol{a})$, we define $T$ to be banded if and only if there exists a positive integer $0 < m < d$ such that $a_s = 0$ for all $s \geq m$.
A banded Toeplitz matrix naturally exhibits sparsity.
Notably, a common approach to estimating sparse matrices is the thresholding procedure, as discussed in \cite{maly2022new, chen2023quantizing}.
In this section, we leverage this technique to develop effective estimation methods tailored to the banded case.

Compared to the low-rank case, the banded case presents additional challenges.
While the results for the low-rank case can be derived with a slight modification of the proof of Theorem~\ref{Theorem_bound_operator_norm}, handling the banded case requires a more intricate approach.
As mentioned before, the sparse case can be addressed using the thresholding method.
Specifically, we consider the thresholded estimator $\breve{T}_{\zeta} := \mathcal{T}_{\zeta}(\hat{T})$ with a given threshold $\zeta$, where
\begin{equation}
        \mathcal{T}_{\zeta}(a) = a \cdot \mathbb{I}(|a| \geq \zeta),
\end{equation}
and $\mathcal{T}_{\zeta}(\hat{T})$ denotes entry-wise thresholding applied to the matrix $\hat{T}$.
In the following, we analyze the error bound of $\breve{T}_{\zeta}$.

We begin by analyzing the estimation error associated with the zero diagonals. 
Specifically, for a Toeplitz matrix $T = \text{Toep}(\boldsymbol{a})$, the ($s+1$)-th diagonal being zero implies that $a_s = 0$. 
By applying Theorem~\ref{Theorem_bound_infty_norm}, we know that if $n \gtrsim \log (d / \delta)$, there exists a constant $C_1 > 0$ such that
\begin{equation}
        \| T - \hat{T} \|_{\infty} \leq C_1 \mathcal{K} \sqrt{\dfrac{\log(d / \delta)}{n}},
\end{equation}
with probability at least $1-\delta$.
Let $\zeta = C_2 \mathcal{K} \sqrt{\dfrac{\log(d / \delta)}{n}}$ with a suitable $C_2 \in (C_1, 2C_1)$ and define $\breve{T}_{\zeta} = \text{Toep}(\breve{\boldsymbol{a}})$.
We consider two cases:
\begin{itemize}
        \item If $|\hat{a}_s| < \zeta$, then $\breve{a}_s = 0$.
        In this case,
        \begin{equation}
                |a_s - \breve{a}_s| = |a_s| 
                \leq |a_s - \hat{a}_s| + |\hat{a}_s| 
                \leq (C_1 + C_2) \mathcal{K} \sqrt{\dfrac{\log(d / \delta)}{n}}.
        \end{equation}

        \item If $|\hat{a}_s| \geq \zeta$, then $\breve{a}_s = \hat{a}_s$. 
        In this case,
        \begin{equation}
                |a_s - \breve{a}_s| = |a_s - \hat{a}_s| \leq C_1 \mathcal{K} \sqrt{\dfrac{\log(d / \delta)}{n}}.
        \end{equation}
        Additionally, note that
        \begin{equation}
                |a_s| \geq |\hat{a}_s| - |a_s - \hat{a}_s| \geq (C_2 - C_1) \mathcal{K} \sqrt{\dfrac{\log(d / \delta)}{n}}.
        \end{equation}
        which implies that
        \begin{equation}
                |a_s - \breve{a}_s| \leq \dfrac{C_1}{C_2 - C_1} |a_s|.
        \end{equation}
\end{itemize}
Combining the above results, we conclude
\begin{equation}
        \label{Eq_relation_a_breve}
        \mathbb{P}\left\{ |a_s - \breve{a}_s| \leq C_3 \min \left(|a_s|,\ \mathcal{K} \sqrt{\dfrac{\log(d / \delta)}{n}}\right)\right\} \geq 1 - \delta,
\end{equation}
where
\begin{equation}
        \label{Eq_relation_c_3}
        C_3 \geq \max\left\{ 2C_2, \dfrac{C_1}{C_2 - C_1} \right\} > 1.
\end{equation}

By leveraging the property \eqref{Eq_relation_a_breve}, the estimation error of $\breve{T}_{\zeta}$ can be effectively controlled.
For a banded Toeplitz matrix $T = \text{Toep}(\boldsymbol{a})$ with $a_s = 0$ for all $s \geq m$, this property significantly reduces the degrees of freedom in the estimation process.
Such a structural simplification enables the threshold estimator $\breve{T}_{\zeta}$ to achieve higher accuracy with fewer samples or lower resolution compared to the general case.
Specifically, we derive a non-asymptotic bound on the estimation error $\|T - \breve{T}_{\zeta}\|_2$ under the operator norm as follows.

\begin{theorem}
        \label{Theorem_banded}
        Under the conditions in Theorem~\ref{Theorem_bound_infty_norm},
        further suppose that $T$ is a bandwidth-$m$ Toeplitz matrix, i.e., $T = \mathrm{Toep}(\boldsymbol{a})$ with $a_s = 0$ for all $s \geq m$.
        For any $\alpha \in [1/2, 1]$ and any $p$ satisfying
        \begin{equation}
                p>1,\quad \text{and}\quad \dfrac{p}{1+2p} \leq \log(d),
        \end{equation}
        consider the ruler $R_\alpha$ and the thresholding estimator $\breve{T}_{\zeta} = \mathcal{T}_{\zeta}(\hat{T})$ with threshold
        \begin{equation}
                \zeta = C \mathcal{K} \sqrt{\dfrac{(2p + 1) \log (d)}{n}},
        \end{equation}
        for some universal constant $C > 0$.
        If $n  \gtrsim (2p + 1) \log (d)$, then with probability at least $1 - e^{-p}$ it holds that
        \begin{equation}
                \label{Eq_bound_band_m_in_theorem}
                \|T - \breve{T}_\zeta \|_2 \leq C m \mathcal{K} \sqrt{\dfrac{\log(d / \delta)}{n}}
        \end{equation}
        where $\mathcal{K} = 4 \|T\|_2 + 8 \Delta^2$.
\end{theorem}

\begin{proof}
        In this proof, we first bound $\mathbb{E} \|T - \breve{T}_{\zeta}\|^p$ using Theorem~\ref{Theorem_bound_infty_norm}, together with the previous analysis. 
        Subsequently, applying Markov's inequality leads directly to the desired result in \eqref{Eq_bound_band_m_in_theorem}.
\end{proof}

The bound established in Theorem~\ref{Theorem_banded} is comparable to that in \cite[Theorem 4]{chen2023quantizing}.
This result indicates that leveraging the band-limited property can further improve the estimation accuracy of the Toeplitz estimator.

\begin{remark}
        Theorem~\ref{Theorem_banded} primarily applies to the scenario where $T$ is known to be band-limited, but \textit{the exact bandwidth $m$ is not specified}.
        In the case that the bandwidth $m$ is known, we can define the estimator as
        \begin{equation}
                \label{BandedE}
                \breve{T} = \text{Toep}(\breve{\boldsymbol{a}}),
        \end{equation}
        where the components of $\breve{\boldsymbol{a}}$ are given by
        \begin{equation}
                \breve{a}_s = \left\{
                        \begin{aligned}
                                &\hat{a}_s,\quad &s < m\\
                                &0,\quad &s \geq m
                        \end{aligned}
                        \right.
        \end{equation}
        The estimator $\breve{T}$ defined in (\ref{BandedE}) also satisfies Theorem~\ref{Theorem_banded}.
\end{remark}

It is worth noting that the results in Theorems~\ref{Theorem_low_rank} and \ref{Theorem_banded} can be naturally extended to the finite-bit setting by appropriately selecting $\Delta$, as discussed in Section~\ref{sec:fini}.

\section{Numerical Results} \label{sec:nume}

We validate our theoretical findings with numerical experiments in this section.
Specifically, we use the Vandermonde decomposition to randomly generate Toeplitz matrices \cite{yang2016vandermonde}.
To be precise, generating a $d$-dimensional Toeplitz matrix involves the following steps for $k \leq d$:
\begin{enumerate}
    \item Sampling $k$ distinct frequencies from $\mathcal{U}([0,1])$ and constructing a Fourier matrix $F_s \in \mathbb{C}^{d \times k}$;

    \item Sampling $k$ independent values from $\mathcal{N}(0,1)$, taking their absolute values as the amplitudes, and then forming a diagonal matrix $D$;

    \item Constructing the complex Toeplitz matrix as $T_c = F_s D F_s^*$.
\end{enumerate}
We take $T = \text{real}(T_c)$ for numerical experiments.
It should be noted that $\text{rank}(T) = \min \{ d, 2k\}$ at this point.
In the subsequent experiments, Toeplitz matrices generated in this manner will be used without further elaboration.

\begin{figure}[htbp]
    \centering
    \includegraphics[scale=0.44]{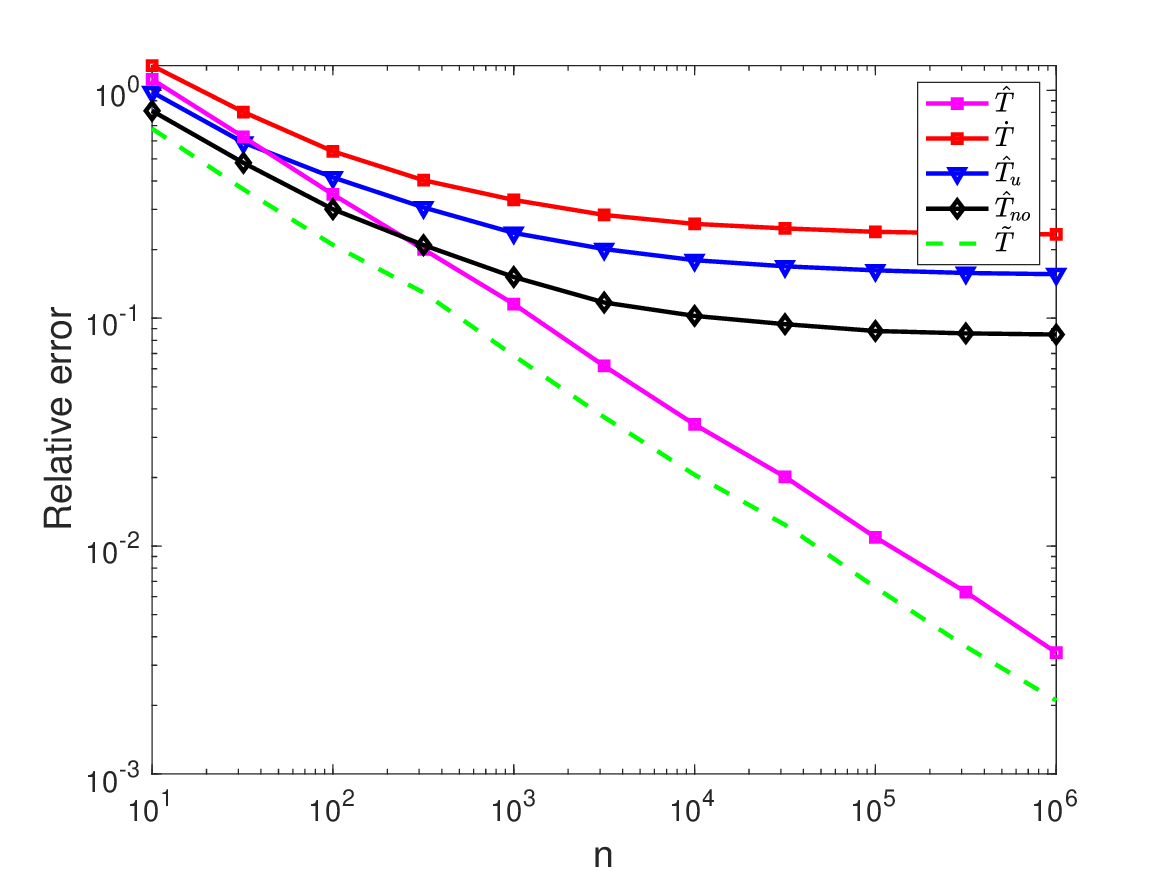}
    \caption{Log-log curve of the relative error $\frac{\| \hat{T} - T\|}{\|T\|}$ with respect to the number of samples $n$}
    \label{LoglogRvN}
\end{figure}

In \textit{Experiment 1}, a symmetric Toeplitz matrix $T \in \mathbb{R}^{16 \times 16}$ is randomly generated.
The quantization level is set to $\Delta = 5$, and the ruler $R_{1/2} = \{1, 2, 3, 4, 8, 12, 16\}$ is employed to estimate $T$ using several estimators:
\begin{itemize}
    \item{The proposed estimator $\hat{T}$: Triangular dither is applied, and $\frac{\Delta^2}{4} I_d$ is subtracted to mitigate the error introduced by the triangular dither.}

    \item{$\dot{T}$: Triangular dither is applied without $\frac{\Delta^2}{4} I_d$ to mitigate the error introduced by the triangular dither.}
    
    \item{$\hat{T}_u$: Uniform dither is applied, and $\frac{\Delta^2}{6} I_d$ is subtracted to mitigate the error introduced by the uniform dither\footnote{
        When uniform dithering is applied, the variance of the corresponding quantization noise $\xi_i$ cannot be computed exactly in theory.
        For simplicity, we \textit{approximately} ignore the correlation between the dither $\tau_i$ and the quantization error $\omega_i$, and \textit{unjustifiably} assume $\mathbb{E}[\xi_i^2] \approx \mathbb{E}[\tau_i^2] + \mathbb{E}[\omega_i^2] = \frac{\Delta^2}{6}$.
        }.}
    
    \item{$\hat{T}_{no}$: Quantization is applied without any dithering.}
    
    \item{$\tilde{T}$: No quantization is applied, serving as the benchmark.}
\end{itemize}
We plot the log-log curve of the relative error versus the number of samples $n$ in Fig.~\ref{LoglogRvN}.
The relative error of the proposed estimator decreases almost linearly with $n$ in the log-log plot.
In contrast, when uniform dither is used, or triangular dither is applied without correction, the relative error approaches a plateau, showing only marginal improvement as the number of samples increases beyond a certain point.
These results demonstrate that the proposed estimator $\hat{T}$ can achieve arbitrary accuracy, even when relying solely on quantized coarse data.
This highlights the superiority of $\hat{T}$ in handling quantized data while maintaining high estimation accuracy.
Moreover, the slope of the curve for $\hat{T}$ is approximately $-0.5007$, in close agreement with the theoretical convergence rate predicted by Theorem~\ref{Theorem_bound_operator_norm}.

\begin{figure}[htbp]
    \centering
    \includegraphics[scale=0.44]{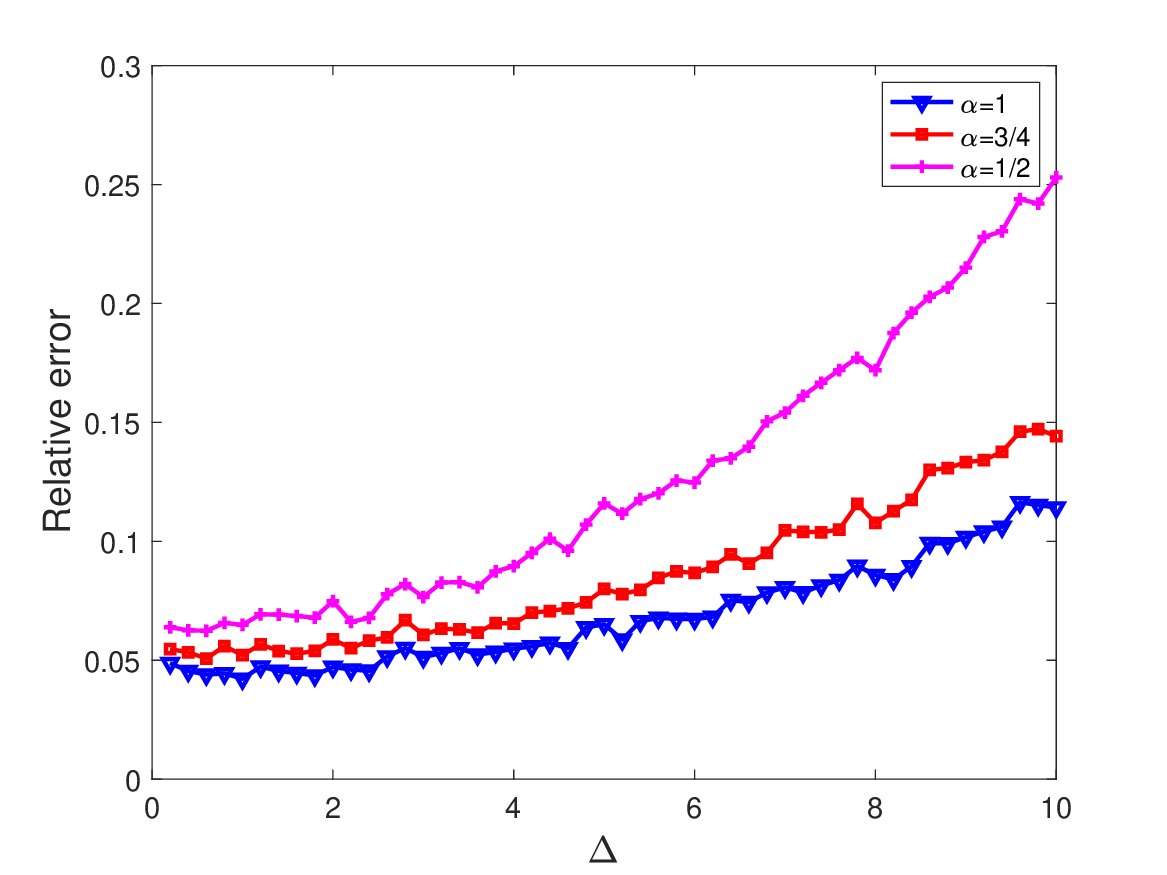}
    \caption{The error variation curves corresponding to different quantization levels $\Delta$}
    \label{PlotRvDel}
\end{figure}

In \textit{Experiment 2}, we further investigate the impact of quantization level $\Delta$ on the estimation error for 16-dimensional symmetric Toeplitz matrices.
Specifically, we fix the number of samples at $n = 1000$ and evaluate the performance of the estimator $\hat{T}$ under three different rulers:
\begin{itemize}
    \item The sparse ruler $R_{1/2} = \{1, 2, 3, 4, 8, 12, 16\}$.

    \item The sparse ruler $R_{3/4} = \{1, 2, 3, 4, 5, 6, 7, 8, 10, 12, 14, 16\}$.

    \item The full ruler $R_1 = \{1, 2, \ldots, 16\}$.
\end{itemize}
The error variation curves corresponding to different quantization levels $\Delta$ are plotted in Fig.~\ref{PlotRvDel}.
The results reveal that as $\Delta$ increases, the error initially varies more gently in the early stages and grows more noticeably in the later stages of the curve.
This trend aligns closely with the theoretical description of $\mathcal{L}$ in our analysis.
These findings indicate that, within a certain range, increasing $\Delta$, i.e., reducing data resolution, has a relatively limited impact on estimation accuracy.
This demonstrates that the estimator maintains robustness against moderate reductions in data resolution, further supporting the feasibility of using quantized data for efficient covariance estimation.

\begin{figure}[htbp]
    \centering
    \includegraphics[scale=0.44]{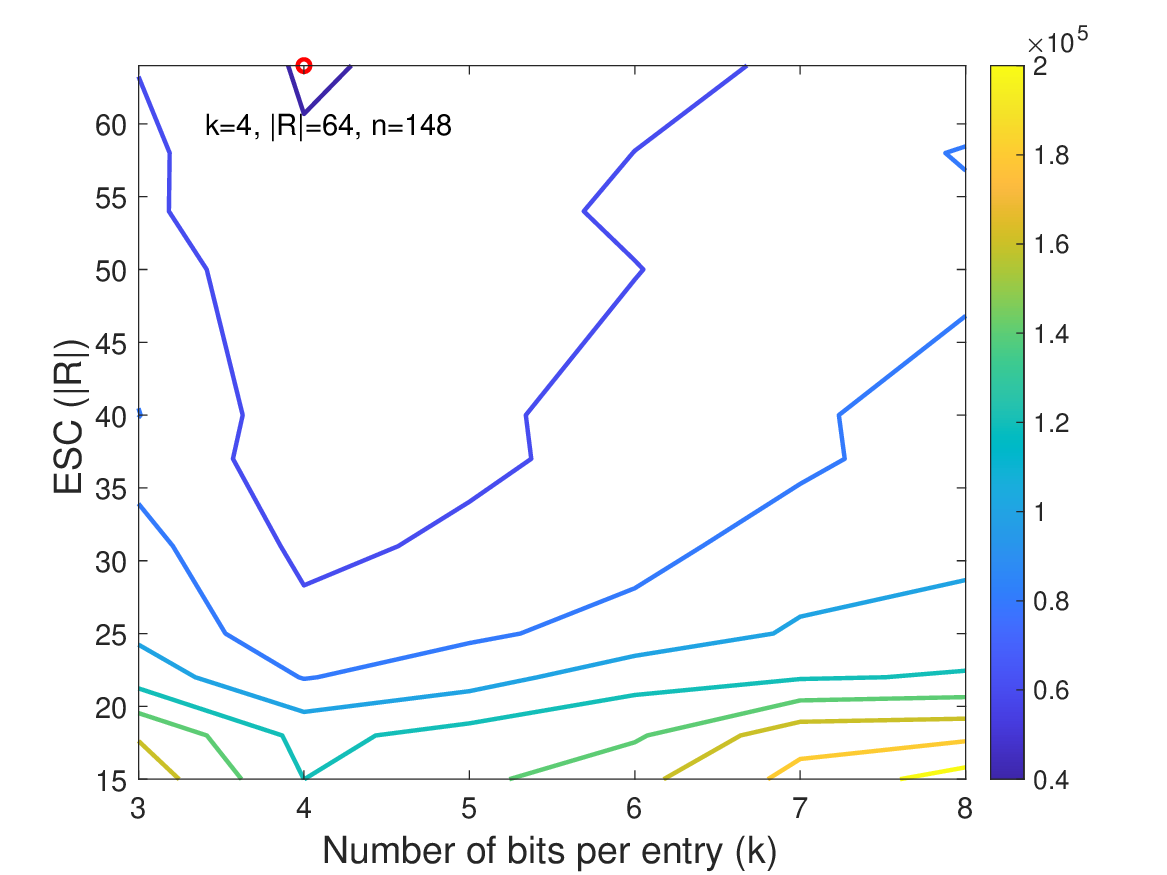}
    \caption{Contour map of the required total number of bits as a function of $k$ and ESC for achieving the preset accuracy $\epsilon = 0.1$}
    \label{PlotESCvsk}
\end{figure}

\begin{figure}[htbp]
    \centering
    \subfloat[VSC v.s. $k$]{
    \includegraphics[scale=0.44]{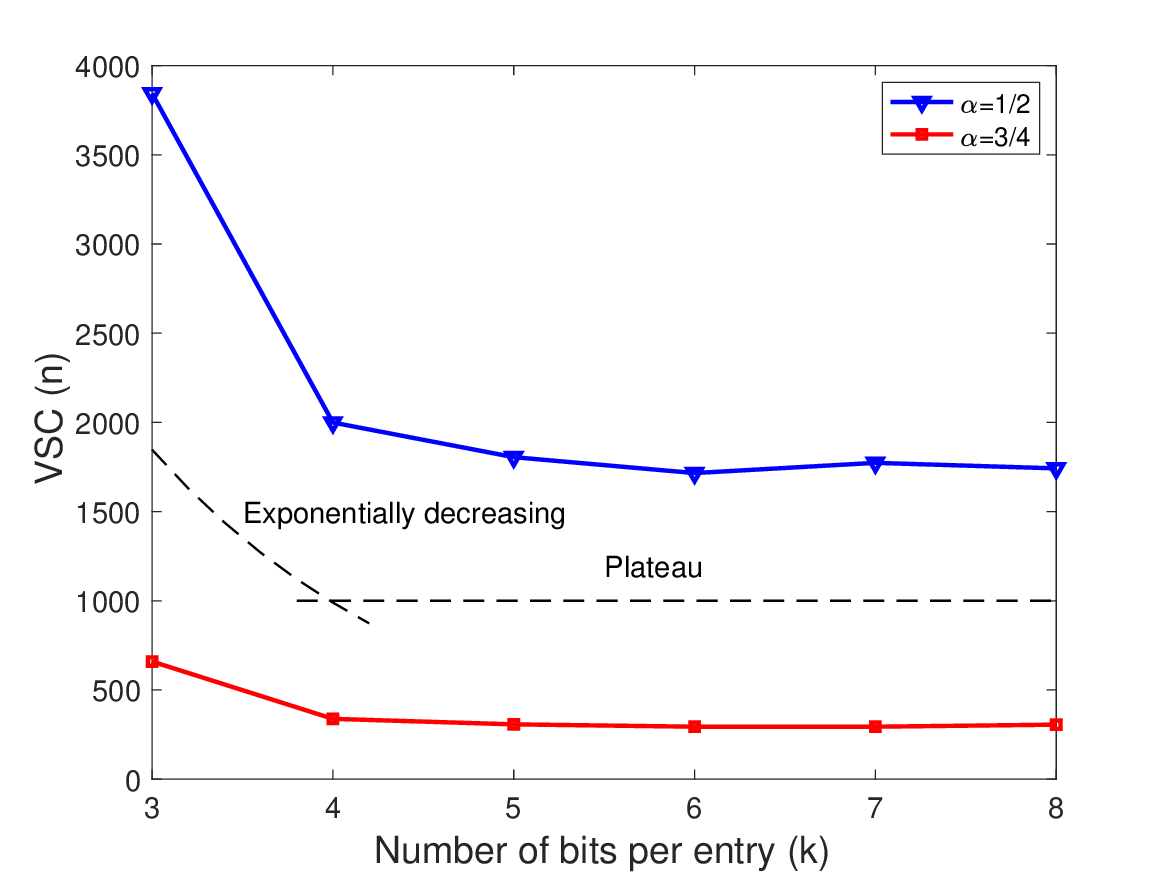}
    }
    \subfloat[VSC v.s. ESC]{
    \includegraphics[scale=0.44]{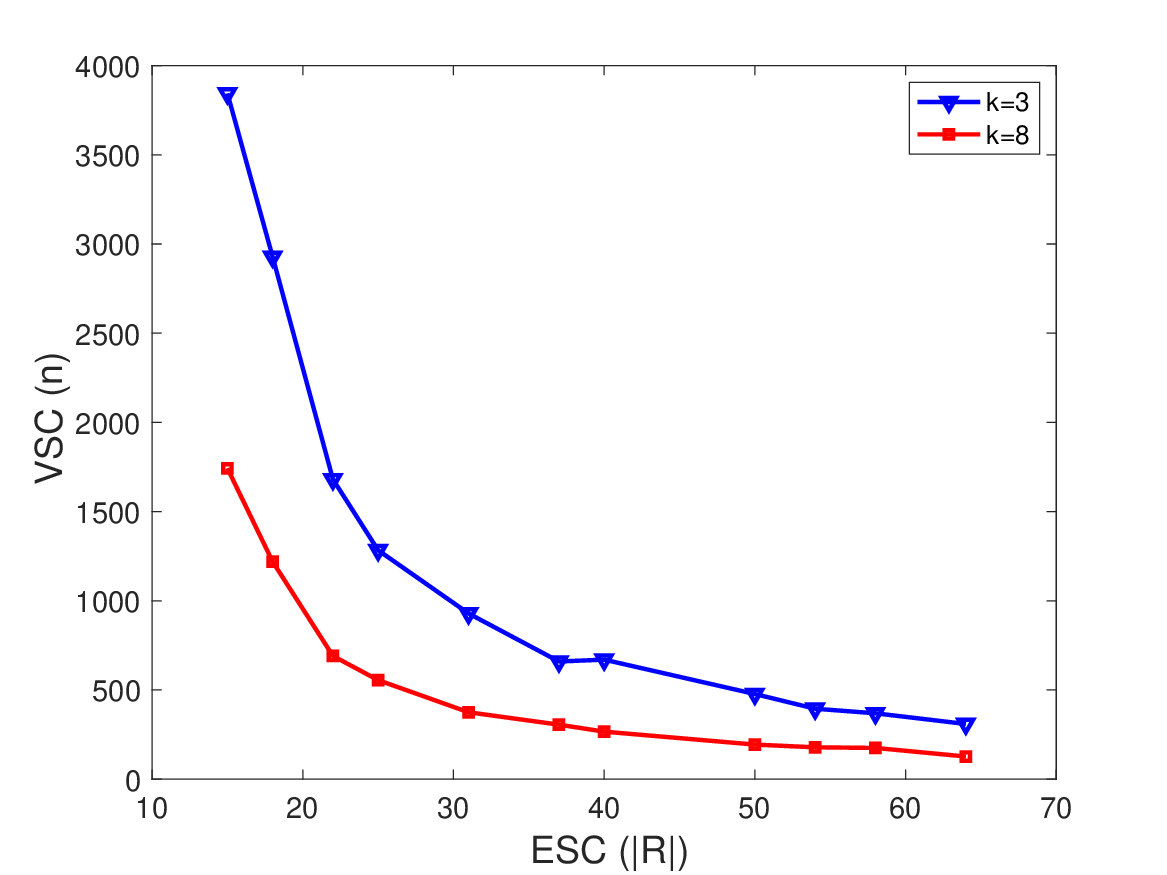}
    }
    \caption{Curve of VSC versus $k$ and ESC to achieve the preset accuracy $\epsilon = 0.1$}
    \label{PlotTradeoff}
\end{figure}

In \textit{Experiment 3}, we randomly generate 64-dimensional symmetric Toeplitz matrices to study how the quantization bit $k$ and the choice of ruler affect estimation performance.
We vary $k$ and $\alpha \in [1/2, 1]$, and for each pair $(k, \alpha)$, we record the minimum sample size $n$ required to achieve $\|T-\hat{T}\|_2 \leq 0.1 \|T\|_2$.
Fig.~\ref{PlotESCvsk} presents a contour map of the required total number of bits as a function of $k$ and ESC.
It can be seen that using the full ruler together with 4-bit quantization per entry minimizes the total number of bits needed for accurate estimation of $T$.
Fig.~\ref{PlotTradeoff} shows several cross-sections of this contour.
For any fixed $\alpha$, when $k$ is small (e.g., $k = 3$), increasing $k$ substantially reduces the required sample size $n$ (i.e., VSC); however, after a certain threshold, further increasing $k$ yields little improvement.
Conversely, for fixed $k$, increasing $\alpha$ (i.e., increasing ESC) clearly decreases the required VSC.
These results validate the detailed trade-off among VSC, ESC, and resolution.

\begin{figure}[htbp]
    \centering
    \includegraphics[scale=0.44]{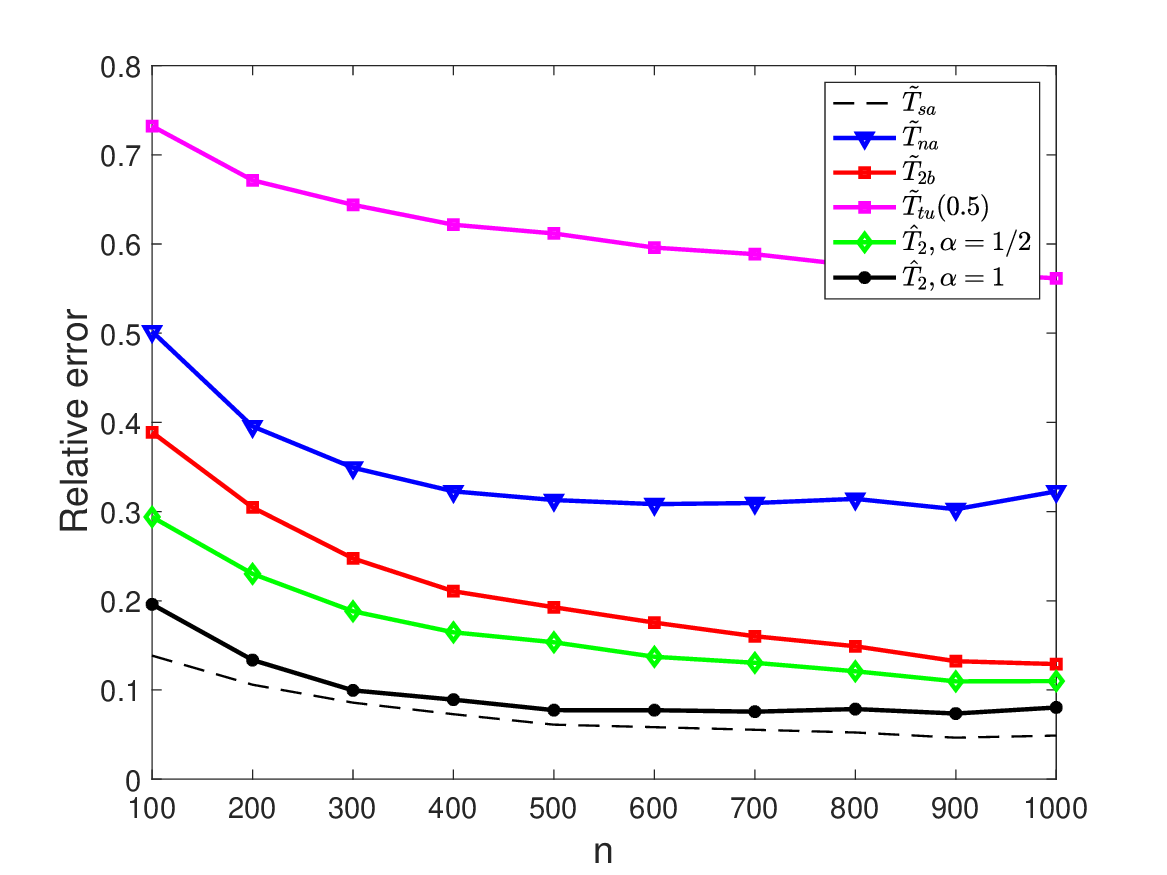}
    \caption{Relative error versus VSC ($n$) for different covariance estimators}
    \label{PlotCompare}
\end{figure}

In \textit{Experiment 4}, we randomly generate 16-dimensional Toeplitz matrices and compare our proposed estimator with several existing 2-bit estimators.
$\tilde{T}_{sa}$ denotes the sample covariance computed from full-resolution data and serves as the benchmark.
$\tilde{T}_{na}$ refers to the estimator in \cite{dirksen2022covariance}.
$\tilde{T}_{2b}$ and $\tilde{T}_{tu}(0.5)$ are both from \cite{chen2025parameter}, where the former is the standard 2-bit estimator and the latter is its parameter-free version with $\lambda = 0.5$.
$\hat{T}_2$ represents our proposed 2-bit estimator, evaluated under different rulers.
All optimal parameters are determined numerically for the case $(n, d) = (500, 16)$.
The results show that in all scenarios, the proposed estimator outperforms the others, even when only partial entries are observed ($\alpha = 1/2$).
This performance gain is attributed to the exploitation of the Toeplitz structure.

\begin{figure}[htbp]
    \centering
    \subfloat[Full rank]{
    \includegraphics[scale=0.44]{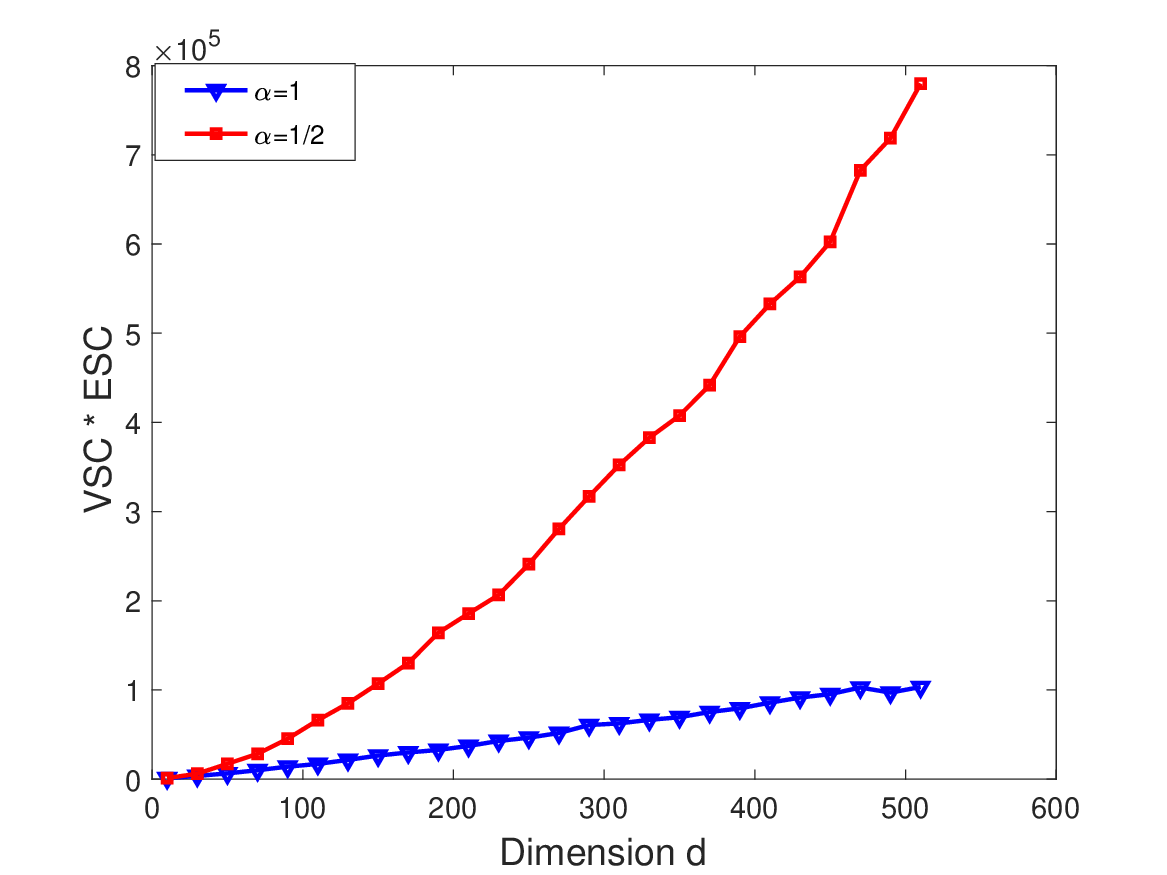}
    }
    \subfloat[$\text{rank}(T) = 10$]{
    \includegraphics[scale=0.44]{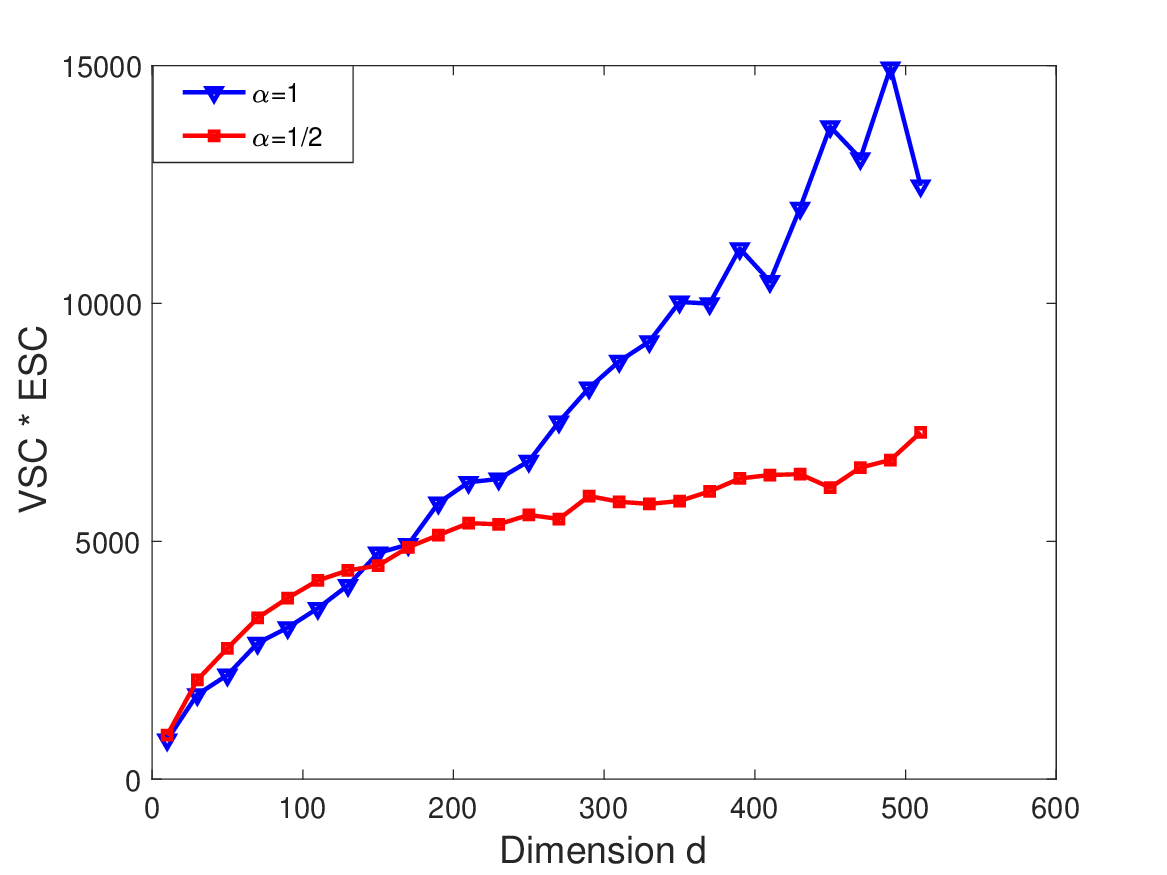}
    }
    \caption{Curve of total sample complexity with respect to data dimensionality required to achieve the preset accuracy $\epsilon = 0.1$}
    \label{PlotRvD}
\end{figure}

In \textit{Experiment 5}, we investigate the trade-off between VSC and ESC under a preset estimation accuracy of $\epsilon = 0.1$ and a fixed quantization level of $\Delta = 2$.
Specifically, we evaluate the total sampling complexity, defined as the product of VSC and ESC, for two different rulers: the sparse ruler $R_{1/2}$ and the full ruler $R_1$.
The variation of total sampling complexity versus matrix dimension is plotted in Fig.~\ref{PlotRvD}.
The results reveal contrasting trends between full-rank and low-rank matrices.
For the full-rank case, while the sparse ruler $R_{1/2}$ effectively reduces the ESC, it leads to a significant increase in VSC, resulting in a sharp rise in total sampling complexity.
In contrast, for low-rank matrices, where the rank is restricted to $10$, the sparse ruler $R_{1/2}$ significantly outperforms the full ruler $R_1$ in terms of total sampling complexity when $d$ is greater than $200$.
This highlights the advantage of using sparse rulers with low ESC in the covariance estimation of low-rank Toeplitz matrices.

\begin{figure}[htbp]
	\centering
	\includegraphics[scale=0.44]{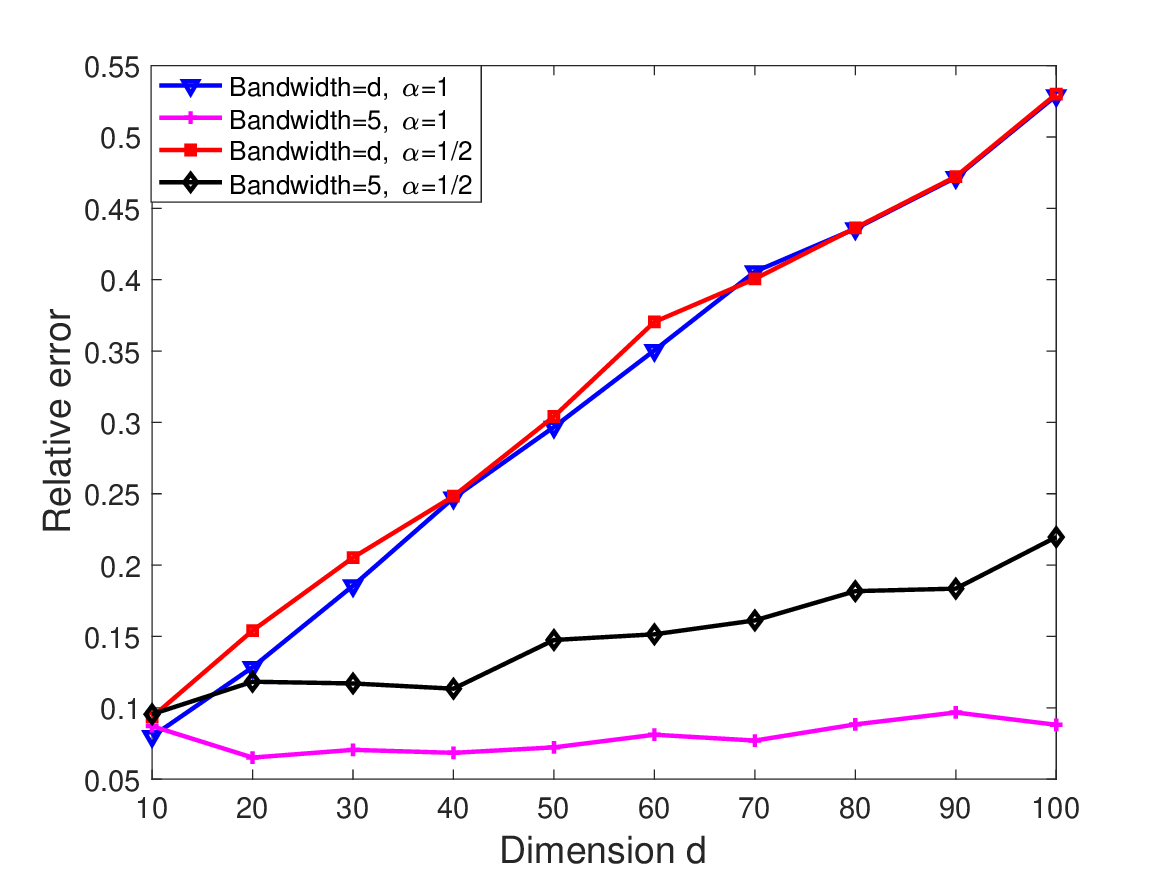}
	\caption{Curve of the relative error versus dimension $d$ for the band-limited case}
	\label{PlotRvD5}
\end{figure}

In \textit{Experiment 6}, we examine the behavior of $\breve{T}_{\zeta}$ as defined in Theorem~\ref{Theorem_low_rank} within the context of a band-limited case.
The sample size is set to $n=1000$, and we assess the performance of the estimator $\breve{T}_{\zeta}$ across various dimensions $d$.
For the band-limited case, the bandwidth is fixed at $m=5$.
As shown in Fig.~\ref{PlotRvD5}, the results indicate that the band-limited case achieves a smaller estimation error for the same number of samples.
Furthermore, once the bandwidth $m$ is specified, the estimation error is essentially independent of $d$.
These findings underscore the effectiveness of $\breve{T}_{\zeta}$ in estimating the band-limited Toeplitz covariance matrix.

Overall, the results of our numerical experiments further validate the effectiveness of the proposed estimator $\hat{T}$.
The experiments confirm that the convergence rate of $\hat{T}$ with respect to the number of samples $n$ is indeed $1/2$, aligning well with the theoretical predictions.
Moreover, the results indicate that reducing the resolution of the data within a certain range has a relatively limited impact on the estimation error, demonstrating the robustness of the estimator under quantized data.
Additionally, the trade-offs between VSC, ESC, and resolution are thoroughly examined through the numerical experiments.
The findings highlight the intricate dependencies among these factors, illustrating how adjustments in one factor can be compensated by changes in others to achieve the desired estimation accuracy.
These observations reinforce the practical value of our proposed estimator in balancing efficiency and accuracy in covariance estimation tasks.

\section{Conclusion} \label{sec:conc}
In this paper, we propose a ruler-based quantized Toeplitz covariance estimator capable of achieving highly accurate estimation of $T$ by observing only a subset of each quantized coarse sample.
We derive non-asymptotic error bounds for the estimator and further analyze the convergence rate  of $\hat{T}$.
Both theoretical analysis and experimental results confirm that the convergence order of the proposed estimator is $1/2$.
Additionally, we demonstrate that reducing the resolution of the data within a certain range has a limited impact on the estimation accuracy.
Our findings highlight the trade-offs between VSC, ESC and resolution, offering a balanced approach to covariance estimation with limited sample information.
These results provide valuable insights and practical guidelines for achieving efficient and accurate covariance estimation in scenarios where data acquisition and processing resources are constrained.

Our exploration is based on Gaussian assumptions, which can be easily generalized to the sub-Gaussian case.
However, extending this work to more general cases, such as heavy-tailed distributions, is a potential direction for future research.
Furthermore, this paper focuses exclusively on Toeplitz covariance estimators with simple computational structures.
While more sophisticated covariance estimators based on advanced algorithms already exist, it would be interesting to investigate how their performance deteriorates under quantization and to explore strategies for mitigating such degradation in practically relevant settings.

\section*{Acknowledgment}
The authors would like to thank Prof. Michael Kwok-Po Ng of Hong Kong Baptist University for helpful comments that have improved the quality of the paper.

\appendix

\subsection{Proof of Theorem~\ref{Theorem_bound_infty_norm}} \label{Proof_bound_infty_norm}
We begin by recalling Bernstein's inequality, which provides a concentration bound for the sum of independent sub-exponential random variables.

\begin{lemma}
        \label{Lemma_bernstein_inequality}
        (Bernstein's inequality, \cite{vershynin2018high})
        Let $X_1, \ldots, X_n$ be independent sub-exponential random variables. Then, for any $t > 0$, it holds that
        \begin{equation}
                \mathbb{P}\left\{\left| \dfrac{1}{n}\sum_{i=1}^{n} X_i - \mathbb{E}(X_i)\right| \geq t \right\}
                \leq 2 \exp{\left[-cn \min \left( \dfrac{t^2}{K^2}, \dfrac{t}{K}\right) \right]},
        \end{equation}
        where $K = \max_{1 \leq i \leq n} \|X_i\|_{\psi_1}$ and $c > 0$ is a constant.
\end{lemma}

Now, we turn to the proof of Theorem~\ref{Theorem_bound_infty_norm}. For any $s \in [d]$, we begin by analyzing the deviation $|a_s - \hat{a}_s|$, which can be expressed as
\begin{equation}
        \label{Eq_deviation_of_a}
        |a_s - \hat{a}_s| 
        = \left|a_s + \dfrac{\Delta^2}{4} \delta_s - \dot{a}_s\right|
        \leq \dfrac{1}{|R_s|} \sum_{(j, k)\in R_s} \left| \dfrac{1}{n} \sum_{l = 1}^n \dot{x}_j^{(l)} \dot{x}_k^{(l)} - \mathbb{E} \left(\dot{x}_j^{(l)} \dot{x}_k^{(l)}\right) \right|.
\end{equation}
To bound \eqref{Eq_deviation_of_a}, we consider any pair $(j, k) \in R_s$. We begin by bounding
\begin{equation*}
        \left| \dfrac{1}{n} \sum_{l = 1}^n  \dot{x}_j^{(l)} \dot{x}_k^{(l)} - \mathbb{E} \left(\dot{x}_j^{(l)} \dot{x}_k^{(l)}\right)\right|.
\end{equation*}
By noting that both $\omega^{(l)}_j$ and $\tau^{(l)}_j$ are bounded random variables, and by applying Hoeffding's lemma\cite{hoeffding1963probability, vershynin2018high}, we obtain
\begin{equation}
        \label{Eq_sub_gaussian_norm_of_x}
        \begin{aligned}
                \| \dot{x}^{(l)}_j \|_{\psi_2}
                = \| x^{(l)}_j + \omega^{(l)}_j + \tau^{(l)}_j \|_{\psi_2}
                \leq \| x^{(l)}_j \|_{\psi_2} + \| \omega^{(l)}_j \|_{\psi_2} + \| \tau^{(l)}_j \|_{\psi_2}
                \leq \sqrt{2 T_{j,j}} + 2\Delta
                \leq \sqrt{2} \|T\|_2^{1/2} + 2\Delta.
        \end{aligned}
\end{equation}
This implies that both $\dot{x}_j^{(l)}$ and $\dot{x}_k^{(l)}$ are sub-Gaussian random variables. Hence, their product $\dot{x}_j^{(l)} \dot{x}_k^{(l)}$ is sub-exponential and satisfies
\begin{equation}
        \label{Eq_sub_exp_norm}
        \begin{aligned}
                 \| \dot{x}_j^{(l)} \dot{x}_k^{(l)} \|_{\psi_1}
                 \leq \| \dot{x}_j^{(l)} \|_{\psi_2} \cdot \| \dot{x}_k^{(l)} \|_{\psi_2}
                 \leq \left(\sqrt{2} \|T\|_2^{1/2} + 2 \Delta\right)^2
                 \leq 4 \|T\|_2 + 8 \Delta^2.
        \end{aligned}
\end{equation}

By Lemma~\ref{Lemma_bernstein_inequality}, there exists a universal constant \( c > 0 \) such that
\begin{equation}
        \begin{aligned}
                \mathbb{P}\left\{ \left| \dfrac{1}{n}\sum_{l=1}^{n} \dot{x}_j^{(l)} \dot{x}_k^{(l)} - \mathbb{E}\left(\dot{x}_j^{(l)} \dot{x}_k^{(l)}\right)\right| \geq t\right\}
                \leq 2 \exp{\left[-cn \min \left( \dfrac{t^2}{\mathcal{K}^2}, \dfrac{t}{\mathcal{K}}\right) \right]},
        \end{aligned}
\end{equation}
where $\mathcal{K} = 4 \|T\|_2 + 8 \Delta^2$. Thus, for \textit{the given index $s$}, we have
\begin{equation}
        \begin{aligned}
                \mathbb{P}\left\{ \dfrac{1}{|R_s|} \sum_{(j, k)\in R_s} \left| \dfrac{1}{n} \sum_{l = 1}^n \dot{x}_j^{(l)} \dot{x}_k^{(l)} - \mathbb{E} \left(\dot{x}_j^{(l)} \dot{x}_k^{(l)}\right) \right| \geq t \right\} 
                &\leq 2 \exp \left[-cn \min \left( \dfrac{t^2}{\mathcal{K}^2}, \dfrac{t}{\mathcal{K}}\right) \right],
        \end{aligned}
\end{equation}
and therefore,
\begin{equation}
        \label{Eq_error_of_a_s}
        \mathbb{P}\left\{ \left| a_s - \hat{a}_s \right| \geq t \right\} \leq 2 \exp{\left[-cn \min \left( \dfrac{t^2}{\mathcal{K}^2}, \dfrac{t}{\mathcal{K}}\right) \right]}.
\end{equation}
We conclude the proof of \eqref{Eq_bound_infty_norm} by applying a union bound over all $s \in [d]$.
The entrywise bound in \eqref{Eq_bound_entrywise} then follows as a direct corollary of \eqref{Eq_bound_infty_norm} by taking $t \asymp \mathcal{K} \sqrt{\dfrac{\log(d / \delta)}{n}}$.

\subsection{Proof of Theorem~\ref{Theorem_bound_operator_norm}} \label{Proof_bound_operator_norm}
As mentioned earlier, we complete the proof by establishing a uniform bound on the spectral density function associated with $T - \hat{T}$.
To that end, we first introduce the Hanson-Wright inequality, which provides a tail bound for quadratic forms of sub-Gaussian random variables and plays a key role in the subsequent analysis.

\begin{lemma}
        \label{Lemma_hanson_wright_inequality}
        (Hanson-Wright inequality, \cite{vershynin2018high})
        Let $\boldsymbol{X} = (X_1, \ldots, X_d)^\top$ be a random vector with independent, zero-mean, sub-Guassian entries, and let $A \in \mathbb{R}^{d \times d}$ be a fixed matrix. Then, for every $t\geq 0$, it holds that
        \begin{equation}
                \mathbb{P}\left\{ \left| \boldsymbol{X}^\top A \boldsymbol{X} - \mathbb{E} [\boldsymbol{X}^\top A \boldsymbol{X} ] \right| > t \right\}
                \leq 2 \exp\left[ -c \min \left( \dfrac{t^2}{K^4 \| A \|_F^2}, \dfrac{t}{K^2 \|A\|_2} \right) \right],
        \end{equation}
        where $K = \max_{1 \leq i \leq n} \|X_i\|_{\psi_2}$, and $c > 0$ is a constant.
\end{lemma}

Then, we proceed to prove Theorem~\ref{Theorem_bound_operator_norm}.
Note that the deviation $T - \hat{T}$ is also a symmetric Toeplitz matrix. Let $\boldsymbol{e} = \boldsymbol{a} - \hat{\boldsymbol{a}}$ denote its associated generating vector.
Since the distribution of $\boldsymbol{x}$ is altered by the quantization process and thus remains unknown, the subsequent analysis is non-trivial.
We divide the proof into the following three steps.

\subsubsection{The non-asymptotic properties of the element-wise error $e_s$}
According to \eqref{Eq_error_of_a_s}, we get a union bound over $s \in [d]$
\begin{equation}
        \label{Eq_union_bound_of_e_s}
        \begin{aligned}
                \mathbb{P}\left\{ \exists s \in [d]: |e_s| \geq t \right\} 
                \leq 2 d \exp{\left[-cn \min \left( \dfrac{t^2}{\mathcal{K}^2}, \dfrac{t}{\mathcal{K}} \right) \right]}
                \leq 2 d \exp{\left(-cn \kappa \right)},
        \end{aligned}
\end{equation}
where $\kappa = \min \left( \dfrac{t^2}{\mathcal{K}^2 \phi(R)}, \dfrac{t}{\mathcal{K} \sqrt{\phi(R)}} \right)$. This indicates that all the coefficients of $L_{\boldsymbol{e}}(x)$ can be effectively controlled within a desired bound.

\subsubsection{Pointwise bound on $L_{\boldsymbol{e}}(x)$}
For a fixed $x \in [0,1]$, we associate the spectral density function $L_{\boldsymbol{e}}(x)$ with a Toeplitz matrix $M\in \mathbb{R}^{d \times d}$, whose entries are given by
\begin{equation}
        M_{j,k} = \dfrac{\cos{(2\pi s x)}}{|R_s|} := M_s, \quad s = |j-k| \in [d].
\end{equation}
In particular, it holds that
\begin{equation}
        L_{\boldsymbol{e}}(x) = e_0 + 2 \sum_{s=0}^{d-1} e_s \cos{(2\pi s x)} 
        = \text{tr}(T_R - \hat{T}_R, M_R).
\end{equation}
By defining the modified sample covariance matrix
\begin{equation}
        \overline{T} = \dfrac{1}{n} \sum_{l=1}^n \dot{\boldsymbol{x}}^{(l)} \dot{\boldsymbol{x}}^{(l)\top} - \dfrac{\Delta^2}{4} I_d,
\end{equation}
the estimator $\hat{T}$ can be interpreted as the Toeplitz projection of $\overline{T}$. That is, for each lag $s \in [d]$,
\begin{equation}
        \sum_{(j,k)\in R_s} \hat{T}_{j,k} = \sum_{(j,k)\in R_s} \overline{T}_{j,k}.
\end{equation}
Therefore,
\begin{equation}
        \begin{aligned}
                \mathrm{tr}(\hat{T}_R, M_R)
                &= \sum_{(j,k) \in R \times R} \hat{T}_{j,k} M_{j,k}
                = \sum_{s=0}^{d-1} \sum_{(j,k)\in R_s} \hat{T}_{j,k} M_s\\
                &= \sum_{s=0}^{d-1} \sum_{(j,k)\in R_s} \overline{T}_{j,k} M_s
                = \sum_{(j,k) \in R \times R} \overline{T}_{j,k} M_{j,k}
                = \mathrm{tr}(\overline{T}_R, M_R).
        \end{aligned}
\end{equation}
Let $\dot{\Sigma} = \mathbb{E}\left( \dot{\boldsymbol{x}}^{(l)} \dot{\boldsymbol{x}}^{(l)\top} \right) = T + \dfrac{\Delta^2}{4} \boldsymbol{I}_d$. 
Then,
\begin{equation}
        \begin{aligned}
                L_{\boldsymbol{e}}(x)
                = \mathrm{tr}(T_R - \hat{T}_R, M_R)
                = \mathrm{tr}(\dot{\Sigma}_R, M_R) - \dfrac{1}{n} \sum_{l=1}^n \dot{\boldsymbol{x}}_R^{(l)\top} M_R \dot{\boldsymbol{x}}_R^{(l)}.
        \end{aligned}
\end{equation}
We can also verify that
\begin{equation}
        \label{Eq_mean_quadratic_form}
        \mathbb{E}\left( \dot{\boldsymbol{x}}_R^{(l)\top} M_R \dot{\boldsymbol{x}}_R^{(l)} \right) = \mathrm{tr}(\dot{\Sigma}_R, M_R).
\end{equation}
Combining \eqref{Eq_mean_quadratic_form} with Lemma~\ref{Lemma_hanson_wright_inequality}, we obtain
\begin{equation}
        \label{Eq_bound_of_l_x}
        \begin{aligned}
                \mathbb{P}\left\{ \left| L_{\boldsymbol{e}}(x) \right| > t \right\}
                &= \mathbb{P}\left\{ \left| \mathrm{tr}(\dot{\Sigma}_R, M_R) - \dfrac{1}{n} \sum_{l=1}^n \dot{\boldsymbol{x}}_R^{(l)\top} M_R \dot{\boldsymbol{x}}_R^{(l)} \right| > t \right\} \\
                &\leq 2 \exp\left[ -cn \min \left( \dfrac{t^2}{(\sqrt{2} \|T\|_2^{1/2} + 2\Delta)^4 \|M_R\|_F^2}, \dfrac{t}{(\sqrt{2} \|T\|_2^{1/2} + 2\Delta)^2 \|M_R\|_2} \right) \right] \\
                &\leq 2 \exp\left[ -cn \min \left( \dfrac{t^2}{\mathcal{K}^2 \phi(R)}, \dfrac{t}{\mathcal{K} \sqrt{\phi(R)}} \right) \right] \\
                & = 2 \exp{\left(-cn \kappa \right)}
        \end{aligned}
\end{equation}
where $\mathcal{K} = 4 \|T\|_2 + 8\Delta^2$ and $\kappa = \min \left( \dfrac{t^2}{\mathcal{K}^2 \phi(R)}, \dfrac{t}{\mathcal{K} \sqrt{\phi(R)}} \right)$.
The last inequality of \eqref{Eq_bound_of_l_x} holds since
\begin{equation}
        \|M_R\|_2^2 \leq \|M_R\|_F^2 \leq \phi(R).
\end{equation}

\subsubsection{Uniform bound on $L_{\boldsymbol{e}}(x)$}
We consider a $\frac{1}{Qd^2}$-net of the interval $[0,1]$ given by $N = \left\{ 0, \frac{1}{Qd^2}, \frac{2}{Qd^2}, \ldots, 1 \right\}$, and set $Q = 8 \pi$.
Similar to \eqref{Eq_union_bound_of_e_s}, we obtain a union bound over the covering net $x \in N$,
\begin{equation}
        \label{Eq_union_bound_of_l_x}
        \mathbb{P}\left\{ \exists x \in N: \left| L_{\boldsymbol{e}}(x) \right| > \dfrac{t}{2} \right\}
        \leq 3Qd^2 \exp{\left(-cn \kappa \right)}.
\end{equation}
Combining \eqref{Eq_union_bound_of_e_s} and \eqref{Eq_union_bound_of_l_x}, we conclude that the following event
\begin{equation}
        \mathcal{A}:\ \  |e_s| \leq t,\  \forall s\in [d],
        \quad  \text{and}\quad 
        |L_{\boldsymbol{e}} (x)| \leq \dfrac{t}{2},\  \forall x\in N
\end{equation}
holds with probability at least $1 - 4 Q d^2 \exp{\left(-cn \kappa \right)}$.
Note that for any $x \in [0,1]$, we can always find an $x' \in N$ such that $0 \leq x - x' < \dfrac{1}{Qd^2}$.
Condition on the event $\mathcal{A}$, we have
\begin{equation}
        |L_{\boldsymbol{e}}(x)|
        \leq |L_{\boldsymbol{e}}(x')| + |x-x'| \sup_{y\in [x', x]}|L'(y)|
        \leq \dfrac{t}{2} + 4 \pi d^2 t \cdot \dfrac{1}{Qd^2} = t,
\end{equation}
where the last inequality holds because
\begin{equation}
        |L'_{\boldsymbol{e}}(y)|
        = \left|4\pi \sum_{s=1}^{d-1} se_s\sin{(2\pi s y)} \right|
        \leq 4 \pi d^2 \|\boldsymbol{e}\|_{\infty} \leq 4 \pi d^2 t.
\end{equation}
Therefore, we obtain the uniform bound
\begin{equation}
        \mathbb{P}\left\{\forall x \in [0,1]: |L_{\boldsymbol{e}} (x)| \leq t \right\} \geq 1 - 32 \pi d^2 \exp{\left(-cn \kappa \right)},
\end{equation}
which is consistent with \eqref{Eq_bound_spectral_density}.

With the help of Lemma~\ref{Lemma_toeplitz_operator_norm_symbol}, the bound in~\eqref{Eq_bound_spectral_density} immediately implies the desired operator norm bound in~\eqref{Eq_bound_spectral_norm_t}.
The bound in \eqref{Eq_bound_operator_t} then follows as a consequence of \eqref{Eq_bound_spectral_norm_t} by taking $t \asymp \mathcal{K} \sqrt{\dfrac{\phi(R) \cdot \log(d / \delta)}{n}}$, thereby completing the proof.

\subsection{Proof of Lemma~\ref{Lemma_bound_phi_r_alpha}} \label{Proof_bound_phi_r_alpha}
The proof follows from elementary counting arguments, as presented in~\cite{eldar2020sample}.  
For completeness, we provide a brief outline here.

On the one hand, for each distance $s \in [0, \lceil d - d^\alpha \rceil]$, there are at least $\lfloor d^{2\alpha - 1} \rfloor$ pairs $(r_1, r_2)$ such that $r_2 - r_1 = s$, where $r_1 \in R_\alpha^{(1)}$ and $r_2 \in R_\alpha^{(2)}$.
This leads to the bound
\begin{equation}
        \label{Eq_ruler_phi_r_alpha_1}
        \sum_{s = 0}^{\lfloor d - d^\alpha \rfloor} \dfrac{1}{|(R_\alpha)_s|} 
        \leq \dfrac{\lfloor d - d^\alpha \rfloor}{\lfloor d^{2 \alpha - 1} \rfloor}
        \leq \dfrac{d}{d^{2 \alpha - 1} / 2}
        \leq 2 d^{2 - 2\alpha}.
\end{equation}
On the other hand, for distances $s$ in the range $d - d^\alpha < s \leq d - 1$, note that
\begin{equation}
        \left\{s \in \mathbb{Z} : d - d^\alpha < s \leq d - 1 \right\} = \bigcup_{j = 1}^{\lceil d^{2\alpha - 1} \rceil} B_j,
\end{equation}
where 
\begin{equation}
        B_j = \left\{ d - (j - 1) d^{1 - \alpha} - 1,\ \ldots,\ \max\{d - j d^{1 - \alpha},\ d - d^\alpha\} \right\}.
\end{equation}
For each $j$, there are at least $j$ elements $r_2' \in R_\alpha^{(2)}$ such that $r_2' > d - (j - 1)d^{1-\alpha} - 1$, implying $|(R_\alpha)_s| \geq j$ for all $s \in B_j$.
This yields the bound
\begin{equation}
        \label{Eq_ruler_phi_r_alpha_2}
        \sum_{s = \lceil d - d^\alpha \rceil} \dfrac{1}{|(R_\alpha)_s|}^{d - 1}
        = \sum_{j=1}^{\lceil d^{2\alpha -1} \rceil} \sum_{s \in B_j} \frac{1}{|(R_\alpha)_s|}
        \leq \sum_{j=1}^{\lceil d^{2\alpha -1} \rceil} \frac{d^{1-\alpha}}{j}
        \leq d^{1-\alpha} \left( 1 + \log\left( \lceil d^{2\alpha -1} \rceil \right) \right).
\end{equation}
The desired bound in~\eqref{Eq_ruler_phi_r_alpha} follows by combining \eqref{Eq_ruler_phi_r_alpha_1} and \eqref{Eq_ruler_phi_r_alpha_2}.

\subsection{Proof of Theorem~\ref{Theorem_lower_bound}} \label{Proof_lower_bound}
We begin by considering the unquantized setting (i.e., $\Delta = 0$), and derive the lower bound by applying Assouad's Lemma and Le Cam's method.
This is a standard technique for establishing minimax lower bounds in covariance estimation and has been widely adopted in the literature (e.g., \cite{cai2013optimal}).

\subsubsection{Assouad-Type Lower Bound}
This bound has been established in \cite{eldar2020sample}.
To ensure that $|T - \hat{T}|_2 \leq \epsilon$ with probability at least $1/10$, it is necessary that
\begin{equation}
        \label{Eq_assouad_type_bound}
        n \gtrsim \dfrac{d^{3 - 4\alpha}}{\epsilon^2}.
\end{equation}
When $\alpha = 1/2$, the lower bound in \eqref{Eq_assouad_type_bound} nearly matches the upper bound in \eqref{Eq_r_alpha_upper_bound}, up to logarithmic factors.
However, for other values of $\alpha$, a non-negligible gap remains between them.

\subsubsection{Le Cam-Type Lower Bound}
We construct two Toeplitz covariance matrices as
\begin{equation}
        T_0 = I_d, \qquad T_1 = I_d + \eta \cdot \mathrm{Toep}(\boldsymbol{\sigma}),
\end{equation}
where the vector $\boldsymbol{\sigma} = (0, 1, 1, \ldots, 1, 0, 0, \ldots, 0)$ consists of $S = \left\lfloor \dfrac{d}{2} \right\rfloor$ consecutive ones starting from the second entry, followed by zeros.
Both $T_0$ and $T_1$ are symmetric positive semidefinite Toeplitz matrices.
We then define the associated zero-mean Gaussian distributions under subsampling by the ruler $R$,
\begin{equation}
        \mathcal{D}_0 = \mathcal{N}(0, (T_0)_R), \qquad \mathcal{D}_1 = \mathcal{N}(0, (T_1)_R).
\end{equation}
We now apply Le Cam's method to derive a lower bound for covariance estimation. 
First, note that
\begin{equation}
        \|T_0 - T_1\|_2 = \eta \|\mathrm{Toep}(\boldsymbol{\sigma})\|_2 \geq \eta S.
\end{equation}
In addition, we observe that the Frobenius norm between the sub-matrices satisfies
\begin{equation}
        \|(T_0)_R - (T_1)_R\|_F^2 
        = \sum_{s=1}^{S} \eta^2 \cdot |R_s| 
        \lesssim \eta^2 \cdot d \cdot d^{\alpha}
        = \eta^2 d^{\alpha+1}.
\end{equation}
Consequently, the KL divergence between the two distributions admits the bound
\begin{equation}
        d_{KL} (\mathcal{D}_0, \mathcal{D}_1)
        \lesssim \left\|I_d - (T_1)_R^{-1/2} (T_0)_R (T_1)_R^{-1/2}\right\|_F^2
        \asymp \|(T_0)_R - (T_1)_R\|_F^2 
        \lesssim \eta^2 d^{\alpha+1}.
\end{equation}
By Pinsker's inequality, this implies
\begin{equation}
        d_{TV} (\mathcal{D}_0^{\otimes n}, \mathcal{D}_1^{\otimes n})
        \lesssim \sqrt{d_{KL} (\mathcal{D}_0^{\otimes n}, \mathcal{D}_1^{\otimes n})}
        = \sqrt{n \cdot d_{KL} (\mathcal{D}_0, \mathcal{D}_1)}
        \lesssim \sqrt{n \eta^2 d^{\alpha+1}}.
\end{equation}
Set $\eta^2 = \dfrac{1}{4 n d^{\alpha+1}}$. Applying Le Cam's lemma \cite{yu1997assouad}, we obtain the lower bound
\begin{equation}
        \label{Eq_sup_e_error_lower}
        \sup_{\hat{T}} \mathbb{E} \left[ \left\|T - \hat{T}\right\|_2 \right] 
        \geq \dfrac{\eta S}{2} \cdot \left(1 - d_{TV} (\mathcal{D}_0^{\otimes n}, \mathcal{D}_1^{\otimes n})\right)
        \gtrsim \dfrac{d}{\sqrt{n d^{\alpha+1}}}
        = \sqrt{\dfrac{d^{1-\alpha}}{n}}.
\end{equation}
Finally, applying Markov's inequality, we conclude that to ensure $\|T - \hat{T}\|_2 \leq \epsilon$ with probability at least $1/10$, one must have
\begin{equation}
        \label{Eq_le_cam_type_bound}
        n \gtrsim \dfrac{d^{1-\alpha}}{\epsilon^2}.
\end{equation}
The bound in \eqref{Eq_le_cam_type_bound} aligns more closely with the upper bound in \eqref{Eq_r_alpha_upper_bound} when when $\alpha = 1$, up to logarithmic factors.

\subsubsection{Unified Lower Bound}
Combining the two lower bounds derived above, we conclude that to guarantee $\|T - \hat{T}\|_2 \leq \epsilon$, the VSC must satisfy
\begin{equation}
        \label{Eq_union_lower_bound}
        n \gtrsim \dfrac{\max(d^{3-4\alpha}, d^{1-\alpha})}{\epsilon^2},
\end{equation}
which aligns with the upper bound in \eqref{Eq_r_alpha_upper_bound} at both endpoints of the range of $\alpha$.

Although the lower bound in \eqref{Eq_union_lower_bound} is derived under the unquantized setting, it also holds for the more challenging quantized case, thereby completing the proof.

\subsection{Proof of Theorem~\ref{Theorem_psd}} \label{Proof_psd}
Following Theorem~\ref{Theorem_bound_operator_norm} and the sample complexity condition in \eqref{Eq_psd_n}, we have
\begin{equation}
        \label{Eq_psd_bound_error}
        \|T - \hat{T}\|_2 
        \leq C \mathcal{K} \sqrt{\dfrac{\phi(R) \cdot \log(d / \delta)}{n}} 
        \leq \dfrac{\lambda_d(T)}{\sqrt{c_0}}
\end{equation}
with probability at least $1 - \delta$.
Applying Weyl's inequality yields
\begin{equation}
        \lambda_d(\hat{T})
        \geq \lambda_d(T) - \|T - \hat{T}\|_2
        \geq \dfrac{\sqrt{c_0} - 1}{\sqrt{c_0}} \lambda_d(T)
        > 0,
\end{equation}
which implies that $\hat{T}$ is positive definite.
Furthermore, by the sub-multiplicative property of the operator norm, we have
\begin{equation}
        \begin{aligned}
                \|T^{-1} - \hat{T}^{-1}\|_2
                &= \|\hat{T}^{-1} (\hat{T} - T) T^{-1}\|_2
                \leq \|\hat{T}^{-1}\|_2 \|T - \hat{T}\|_2 \|T^{-1}\|_2 \\
                &= \dfrac{1}{\lambda_d(\hat{T})} \|T - \hat{T}\|_2 \dfrac{1}{\lambda_d(T)}
                \leq \dfrac{C \sqrt{c_0}}{\sqrt{c_0} - 1} \cdot \dfrac{\mathcal{K}}{\lambda_d^2(T)} \sqrt{\dfrac{\phi(R) \cdot \log(d / \delta)}{n}}.
        \end{aligned}
\end{equation}
This completes the proof.

\subsection{Proof of Theorem~\ref{Theorem_finite_bit_quant}} \label{Proof_finite_bit_quant}
As mentioned earlier, we first establish that $\hat{T}_k = \hat{T}$ holds with high probability. 
We observe that $\hat{T}_k = \hat{T}$ holds if
\begin{equation}
        \mathcal{Q}_{\Delta} \left( x_{j}^{(l)} + \tau_{j}^{(l)} \right) = \mathcal{Q}_{\Delta, k} \left( x_{j}^{(l)} + \tau_{j}^{(l)} \right), \quad \forall j \in R,\ \forall l = 1, \ldots, n,
\end{equation}
which is guaranteed if
\begin{equation}
        \label{Eq_condition_on_delta_pre}
        2^{k-1} \Delta > \max_{j \in R,\ 1 \leq l \leq n} \left| x_j^{(l)} + \tau_j^{(l)} \right|.
\end{equation}
Since $\tau_j^{(l)}$ is supported on $\left[-\Delta/2, \Delta/2\right]$, the condition \eqref{Eq_condition_on_delta_pre} further implies
\begin{equation}
        \max_{j \in R,\ 1 \leq l \leq n} \left| x_j^{(l)} \right| < \left(2^{k-1} - 1\right) \Delta.
\end{equation}

Note that $\boldsymbol{x}^{(l)} \sim \mathcal{N}(0, T)$.
By applying the Gaussian tail bound (see e.g.,~\cite{vershynin2018high}), we have
\begin{equation}
        \mathbb{P}\left\{ \left|x_j^{(l)}\right| \geq t\right\} 
        \leq 2 \exp{\left( -\dfrac{t^2}{2 T_{jj}} \right)} \leq 2 \exp{\left( -\dfrac{t^2}{2 \|T\|_{\infty}} \right)}.
\end{equation}
Taking a union bound over $j \in R$ and $l = 1, \ldots, n$, we obtain
\begin{equation}
        \label{Eq_finite_bit_union_bound_x}
        \mathbb{P}\left\{ \max_{j\in R, 1 \leq l \leq n} \left|x_j^{(l)}\right| \geq t\right\} 
        \leq 2 n |R| \exp{\left( -\dfrac{t^2}{2 \|T\|_{\infty}} \right)}.
\end{equation}
If we set
\begin{equation}
        \Delta = C_{\mathrm{bit}} \cdot 2^{-k} \sqrt{ \|T\|_{\infty} \log\left( \frac{2n|R|}{\delta'} \right)}
\end{equation}
with $C_{bit} \geq \sqrt{2}$, and choose $t = (2^{k-1} - 1) \Delta$ in \eqref{Eq_finite_bit_union_bound_x}, we have
\begin{equation}
        \label{Eq_finite_bit_bound_t_hat_k}
        \mathbb{P}\left\{ \hat{T}_k \neq \hat{T} \right\} 
        \leq \mathbb{P}\left\{ \max_{j\in R, 1 \leq l \leq n} \left|x_j^{(l)}\right| \geq (2^{k-1} - 1) \Delta\right\} 
        \leq \delta'.
\end{equation}
Combining~\eqref{Eq_finite_bit_bound_t_hat_k} with Theorem~\ref{Theorem_bound_operator_norm}, we conclude that if $n \gtrsim \log (d / \delta)$, it holds with probability at least $1 - \delta - \delta'$ that
\begin{equation}
        \|T - \hat{T}_k\|_2 \leq C \mathcal{K} \sqrt{\dfrac{\phi(R) \cdot \log(d / \delta)}{n}}.
\end{equation}
The proof is thus complete.

\subsection{Proof of Theorem~\ref{Theorem_low_rank}} \label{Proof_low_rank}
For all $j, k \in R_{\alpha}$, it holds that $\left|T_{j,k}\right| \leq \|T_{R_{\alpha}}\|_2$.
Furthermore, the sub-Gaussian norm bound in \eqref{Eq_sub_gaussian_norm_of_x} can be refined as
\begin{equation}
        \|\dot{x}_j^{(l)}\|_{\psi_2} \leq \sqrt{2} \|T_{R_{\alpha}}\|_2^{1/2} + 2\Delta,
\end{equation}
Consequently, applying Lemma~\ref{Lemma_low_rank_bound}, the sub-exponential norm bound in~\eqref{Eq_sub_exp_norm} can be sharpened to
\begin{equation}
        \|\dot{x}_j^{(l)} \dot{x}_k^{(l)}\|_{\psi_1} \leq 4 \|T_{R_{\alpha}}\|_2 + 8 \Delta^2 \leq \dfrac{24 r}{d^{1-\alpha}} \|T\|_2  + 8 \Delta^2.
\end{equation}
Following the derivations in the proof of Theorem~\ref{Theorem_bound_operator_norm} (see Appendix~\ref{Proof_bound_operator_norm}), we obtain
\begin{equation}
        \label{Eq_bound_low_rank_bound_in_proof}
        \mathbb{P}\left\{\|T - \hat{T}\|_2 \geq t \right\} \leq 32 \pi d^2 \exp{\left(-cn \kappa_r \right)},
\end{equation}
where $\kappa_r = \min \left( \dfrac{t^2}{\mathcal{K}_r^2 \phi(R_\alpha)}, \dfrac{t}{\mathcal{K}_r \sqrt{\phi(R_\alpha)}} \right)$ and $\mathcal{K}_r = \dfrac{24 r}{d^{1-\alpha}} \|T\|_2 + 8 \Delta^2$.
Finally, combining \eqref{Eq_bound_low_rank_bound_in_proof} with Lemma~\ref{Lemma_bound_phi_r_alpha}, and setting $t \asymp \mathcal{K}_r \sqrt{\dfrac{\phi(R_\alpha) \cdot \log(d / \delta)}{n}}$, we establish the desired bound \eqref{Eq_low_rank}, which completes the proof.

\subsection{Proof of Theorem~\ref{Theorem_banded}} \label{Proof_banded}
This proof is partially inspired by the techniques presented in \cite{chen2023quantizing}, but incorporates several technical refinements and detailed improvements.
Define the event
\begin{equation}
        \mathcal{B}_s: \quad |a_s - \breve{a}_s| \leq C_3 \min \left(|a_s|, \mathcal{K} \sqrt{\dfrac{\log(d / \delta)}{n}}\right).
\end{equation}
Our objective is to bound $\mathbb{E} \|T - \breve{T}_{\zeta}\|^p$.
Observe that
\begin{equation}
        \begin{aligned}
                \mathbb{E} \|T - \breve{T}_{\zeta}\|^p \leq& \mathbb{E} \left(\sum_{s=0}^{d-1} |a_s - \breve{a}_s| \mathbb{I}(\mathcal{B}_s) + \sum_{s=0}^{d-1} |a_s - \breve{a}_s| \mathbb{I}(\mathcal{B}_s^C) \right)^p \\
                \leq& 2^p \mathbb{E} \left(\sum_{s=0}^{d-1} |a_s - \breve{a}_s| \mathbb{I}(\mathcal{B}_s)\right)^p + 2^p \mathbb{E} \left(\sum_{s=0}^{d-1} |a_s - \breve{a}_s| \mathbb{I}(\mathcal{B}_s^C)\right)^p \\
                :=& 2^p (I_1 + I_2),
        \end{aligned}
\end{equation}
where $I_1$ and $I_2$ correspond to the contributions from the events $\mathcal{B}_s$ and its complement, respectively.
In what follows, we bound $I_1$ and $I_2$ separately.

\subsubsection{Bounding $I_1$}
When the event $\mathcal{B}_s$ occurs (i.e., $\mathbb{I}(\mathcal{B}_s) = 1$), we know that if $a_s = 0$, then $\breve{a}_s = 0$.
Under the bandwidth assumption on $T = \mathrm{Toep}(\boldsymbol{a})$, we have
\begin{equation}
        \label{Eq_banded_bound_i_1}
        I_1 \leq
        \mathbb{E} \left(\sum_{s=0}^{m-1} |a_s - \breve{a}_s| \mathbb{I}(\mathcal{B}_s)\right)^p 
        \leq \left(C_3 m \mathcal{K} \sqrt{\dfrac{\log(d / \delta)}{n}}\right)^p.
\end{equation}

\subsubsection{Bounding $I_2$}
Next, consider the case where $\mathcal{B}_s$ does not hold, i.e., $\mathbb{I}(\mathcal{B}_s^C) = 1$.
In this scenario, observe that
\begin{equation}
        \begin{aligned}
                \left(\sum_{s=0}^{d-1} |a_s - \breve{a}_s| \mathbb{I}(\mathcal{B}_s^C)\right)^p 
                &\leq \left(\sum_{s=0}^{d-1} |a_s - \hat{a}_s| \mathbb{I}(\mathcal{B}_s^C) + \sum_{s=0}^{d-1} |\hat{a}_s - \breve{a}_s| \mathbb{I}(\mathcal{B}_s^C)\right)^p \\
                &= \left(\sum_{s=0}^{d-1} |a_s - \hat{a}_s| \mathbb{I}(\mathcal{B}_s^C) + \sum_{s=0}^{d-1} |\hat{a}_s - \breve{a}_s| \mathbb{I}(|\hat{a}_s| < \zeta) \mathbb{I}(\mathcal{B}_s^C)\right)^p \\
                &\leq (2d)^p \left(\sum_{s=0}^{d-1} |a_s - \hat{a}_s|^p \mathbb{I}(\mathcal{B}_s^C) + \sum_{s=0}^{d-1} |\hat{a}_s|^p \mathbb{I}(|\hat{a}_s| < \zeta) \mathbb{I}(\mathcal{B}_s^C)\right).
        \end{aligned}
\end{equation}
Therefore, the term $I_2$ can be bounded as
\begin{equation}
        \label{Eq_banded_bound_i_2}
        I_2
        \leq (2d)^p \left(\sum_{s=0}^{d-1} \mathbb{E}\left[|a_s - \hat{a}_s|^p \mathbb{I}(\mathcal{B}_s^C)\right]
        + \sum_{s=0}^{d-1} \mathbb{E}\left[|\hat{a}_s|^p \mathbb{I}(|\hat{a}_s| < \zeta) \mathbb{I}(\mathcal{B}_s^C)\right]\right) := (2d)^p (I_{21} + I_{22}).
\end{equation}

\textit{2.1. Bounding $I_{21}$:} By the Cauchy–Schwarz inequality, we have
\begin{equation}
        \mathbb{E}\left[|a_s - \hat{a}_s|^p \mathbb{I}(\mathcal{B}_s^C)\right] 
        \leq \sqrt{\mathbb{E}\left[|a_s - \hat{a}_s|^{2p}\right] \cdot \mathbb{E}[\mathbb{I}(\mathcal{B}_s^C)]}
        = \sqrt{\mathbb{E}\left[|a_s - \hat{a}_s|^{2p}\right] \cdot \mathbb{P}(\mathcal{B}_s^C)}.
\end{equation}
Applying the tail integration formula and using the concentration result in \eqref{Eq_error_of_a_s}, we obtain
\begin{equation}
        \label{Eq_banded_bound_i_21_1}
        \begin{aligned}
                \mathbb{E}\left[|a_s - \hat{a}_s|^{2p}\right]
                &= 2p \int_0^{\infty} t^{2p-1} \mathbb{P}\left\{ \left| a_s - \hat{a}_s \right| \geq t \right\} \text{d} t\\
                &\leq 4p \int_0^{\infty} t^{2p-1} \exp{\left[-cn \min \left( \dfrac{t^2}{\mathcal{K}^2}, \dfrac{t}{\mathcal{K}}\right) \right]} \text{d} t\\
                &= 4p \int_0^{\mathcal{K}} t^{2p-1} \exp{\left(-cn \dfrac{t^2}{\mathcal{K}^2} \right)} \text{d} t + 4p \int_{\mathcal{K}}^{\infty} t^{2p-1} \exp{\left(-cn \dfrac{t}{\mathcal{K}} \right)} \text{d} t\\
                &\leq 4p \int_0^{\infty} t^{2p-1} \exp{\left(-cn \dfrac{t^2}{\mathcal{K}^2} \right)} \text{d} t + 4p \int_{0}^{\infty} t^{2p-1} \exp{\left(-cn \dfrac{t}{\mathcal{K}} \right)} \text{d} t\\
                &\leq \dfrac{2p\mathcal{K}^{2p}}{(cn)^p} \cdot \Gamma (p) + \dfrac{4p\mathcal{K}^{2p}}{(cn)^{2p}} \cdot \Gamma (2p)\\
                &\leq \dfrac{2\mathcal{K}^{2p}}{(cn)^p} \cdot p^p + \dfrac{2\mathcal{K}^{2p}}{(cn)^{2p}} \cdot (2p)^{2p},
        \end{aligned}
\end{equation}
where the last equation holds because $\Gamma(p+1) \leq p^p$.
Combining~\eqref{Eq_banded_bound_i_21_1} with the fact that $\mathbb{P}(\mathcal{B}_s^C) \leq \delta$, we further obtain
\begin{equation}
        \label{Eq_banded_bound_i_21_2}
        \begin{aligned}
                I_{21}
                \leq \sqrt{\left( \dfrac{2\mathcal{K}^{2p}}{(cn)^p} \cdot p^p + \dfrac{2\mathcal{K}^{2p}}{(cn)^{2p}} \cdot (2p)^{2p} \right) \cdot \delta} 
                \leq \sqrt{2 \delta} \cdot \mathcal{K}^p \cdot \sqrt{\dfrac{p^p}{(cn)^p} + \dfrac{(2p)^{2p}}{(cn)^{2p}}}. 
        \end{aligned}
\end{equation}
Set $\delta = d^{-2p}$ and choose $p$ such that
\begin{equation}
        p>1,\quad \text{and}\quad \dfrac{p}{1+2p} \leq \log(d),
\end{equation}
Substituting this into~\eqref{Eq_banded_bound_i_21_2}, we obtain
\begin{equation}
        \label{Eq_banded_bound_i_21}
        I_{21} \leq \sqrt{2} d^{-p} \cdot \left( C_4 \mathcal{K} \sqrt{\dfrac{\log(d / \delta)}{n}} \right)^p
\end{equation}
for some $C_4 > 0$.

\textit{2.2. Bounding $I_{22}$:} Observe that $|\hat{a}_s| < \zeta$ implies $\breve{a}_s = 0$, and the event $\mathcal{B}_s^C$ further implies
\begin{equation}
        \label{Eq_banded_bound_i_22_1}
        |a_s| = |a_s - \breve{a}_s| > C_3 \mathcal{K} \sqrt{\dfrac{\log(d / \delta)}{n}} \geq 2 \zeta 
        > 2 |\hat{a}_s| \geq 2|a_s| - 2|a_s - \hat{a}_s|.
\end{equation}
Therefore,
\begin{equation}
        \dfrac{1}{2} |a_s| < |a_s - \tilde{a}_s|,
\end{equation}
which implies
\begin{equation}
        \label{Eq_banded_bound_i_22_2}
        I_{22} = \mathbb{E}\left[|\hat{a}_s|^p \mathbb{I}(|\hat{a}_s| < \zeta) \mathbb{I}(\mathcal{B}_s^C)\right] 
        \leq \zeta^p \mathbb{P} \left\{|a_s - \hat{a}_s| > \dfrac{1}{2} |a_s|\right\}.
\end{equation}
Since \eqref{Eq_banded_bound_i_22_1} guarantees $|a_s| > 2\zeta$, applying Theorem~\ref{Theorem_bound_infty_norm} gives
\begin{equation}
        \label{Eq_banded_bound_i_22_3}
        \begin{aligned}
                \mathbb{P} \left\{|a_s - \hat{a}_s| > \dfrac{1}{2} |a_s|\right\}
                \leq& 2 \exp \left[ -cn \min \left( \dfrac{|a_s|^2}{4\mathcal{K}^2}, \dfrac{|a_s|}{2\mathcal{K}}\right)\right]
                = 2 \exp \left[ -cn \dfrac{\zeta^2}{\mathcal{K}^2}\right]\\
                \leq &2 \exp \left[-cC_2^2 \log(d / \delta) \right]
                \leq 2 d^{-2p-1},
        \end{aligned}
\end{equation}
where the last inequality holds by choosing $C_2$ such that $c C_2^2 \geq 1$.
Combining \eqref{Eq_banded_bound_i_22_2} and \eqref{Eq_banded_bound_i_22_3}, we obtain
\begin{equation}
        \label{Eq_banded_bound_i_22}
        I_{22} \leq 2 \zeta^p d^{-2p-1} \leq 2 d^{-2p} \left( C_2 \mathcal{K} \sqrt{\dfrac{\log(d / \delta)}{n}} \right)^p.
\end{equation}

With the bounds in \eqref{Eq_banded_bound_i_1}, \eqref{Eq_banded_bound_i_21} and \eqref{Eq_banded_bound_i_22}, we conclude that
\begin{equation}
        \begin{aligned}
                \mathbb{E} \|T - \breve{T}_{\zeta}\|^p 
                &\leq \left(2 C_3 m \mathcal{K} \sqrt{\dfrac{\log(d / \delta)}{n}}\right)^p 
                + \sqrt{2} \cdot \left( 4 C_4 \mathcal{K} \sqrt{\dfrac{\log(d / \delta)}{n}} \right)^p
                + 2 d^{-p} \left( 4 C_2 \mathcal{K} \sqrt{\dfrac{\log(d / \delta)}{n}} \right)^p\\
                &\leq \left(C_5 m \mathcal{K} \sqrt{\dfrac{\log(d / \delta)}{n}}\right)^p,
        \end{aligned}
\end{equation}
for some universal constant $C_5 > 0$.
Applying Markov's inequality yields
\begin{equation}
        \mathbb{P}\left\{\|T - \breve{T}_\zeta \|_2 \geq C_5 e m \mathcal{K} \sqrt{\dfrac{\log(d / \delta)}{n}}\right\}
        \leq \dfrac{\mathbb{E} \|T - \breve{T}_{\zeta}\|_2^p}{\left(C_5 e m \mathcal{K} \sqrt{\frac{\log(d / \delta)}{n}}\right)^p} = e^{-p}.
\end{equation}
This completes the proof.

\bibliographystyle{IEEEtran}
\bibliography{QTCE_reference}

\end{document}